\documentclass[11pt]{article}

\usepackage[T1]{fontenc}
\usepackage{lmodern}
\usepackage{xcolor}
\usepackage[margin=1in]{geometry}
\usepackage{amsmath,amssymb,amsthm,mathtools}
\usepackage{microtype}
\usepackage[hidelinks]{hyperref}
\usepackage{authblk}

\newtheorem{theorem}{Theorem}[section]
\newtheorem{lemma}[theorem]{Lemma}

\newtheorem{proposition}[theorem]{Proposition}
\newtheorem{corollary}[theorem]{Corollary}
\theoremstyle{definition}
\newtheorem{definition}[theorem]{Definition}
\theoremstyle{remark}

\DeclareMathOperator{\Tr}{Tr}
\DeclareMathOperator{\supp}{supp}

\DeclareMathOperator{\QIC}{QIC}
\DeclareMathOperator{\QCC}{QCC}
\DeclareMathOperator{\AND}{AND}
\DeclareMathOperator{\DISJ}{DISJ}
\DeclareMathOperator{\id}{id}
\newcommand{\ket}[1]{\lvert #1\rangle}
\newcommand{\bra}[1]{\langle #1\rvert}

\newcommand{\norm}[1]{\lVert #1\rVert}
\newcommand{\abs}[1]{\lvert #1\rvert}

\newcommand{\proj}[1]{\ket{#1}\!\bra{#1}}

\newcommand{\eps}{\varepsilon}

\title{Information Potential: A Variational Approach to Quantum Information Complexity}
\author[1,2]{Penghui Yao\thanks{phyao1985@gmail.com}}
\author[1]{Yifan Zhou\thanks{zyfzhouyifan@outlook.com}}

\affil[1]{State Key Laboratory of Novel Software Technology, Nanjing University, Nanjing 210023, China}
\affil[2]{Hefei National Laboratory, Hefei 230088, China}

\date{}

\begin{document}
\maketitle

\begin{abstract}
Quantum information complexity (QIC), introduced by Touchette [Touchette, STOC 2015], has been shown to be one of the most powerful methods for proving quantum communication complexity and has also been shown to be equal to amortized quantum communication complexity. Unfortunately, QIC is generally hard to analyze because it is a sum of quantum conditional mutual information terms, which are difficult to estimate. In this work, we introduce a new variational approach, the information potential, for analyzing QIC and proving quantum communication lower bounds inspired by the resolvent representation for quantum relative entropy. This approach connects the QIC of individual messages to a quadratic form, making it more amenable and thus enables us to bound the cumulative positive increments of the potential throughout an interactive quantum protocol by its QIC. Lower bounds on the growth of the potential therefore translate into lower bounds on both QIC and quantum communication complexity.

As an application, we give an optimal $\Omega(1/r)$ lower bound on the QIC of the two-bit $\mathsf{AND}$ function as well as an optimal $\Omega(n/r)$ lower bound on the quantum communication complexity of $r$-round Set-Disjointness, answering an open problem in~[Braverman, Garg, Ko, Mao, Touchette FOCS 2015]. Moreover, we further prove a direct-sum theorem for bounded-round quantum communication complexity of Set Disjointness. With the tight bound on the QIC of $\mathsf{AND}$ function, we further establish a nearly tight tradeoff for the asymmetric quantum communication complexity of $\mathsf{Set} \mathsf{Disjointness}$:
$(q_A+1)(q_B+1)=\Omega(n)$,
where $q_A$ and $q_B$ denote the total numbers of qubits sent by Alice and Bob, respectively.
\end{abstract}

\tableofcontents

\section{Introduction}
\label{sec:introduction}

Quantum information complexity (QIC) measures the amount of information that is exchanged between the parties during an interactive quantum communication protocol, rather than the total number of qubits transmitted. It was introduced by Touchette~\cite{Tou15} as a fully quantum counterpart of the classical notion of information complexity (IC)~\cite{CSWY01,BYJKS04,BBCR10}, which has played a central role in the study of randomized communication complexity~\cite{BBCR10,BR11,Bra12,BGPW13}. While communication complexity charges the total number of bits exchanged throughout a communication protocol, IC measures only the information about the inputs revealed through the communication. This distinction makes IC particularly useful for understanding how much of the communication in a protocol is intrinsically necessary and how much can, in principle, be compressed. In the quantum setting, QIC enjoys analogous fundamental properties: it is one of the most powerful lower bound methods for quantum communication complexity, and satisfies natural information-theoretic properties such as additivity under independent instances~\cite{Tou15}. Most importantly, Touchette established an operational characterization of QIC by showing that it coincides with amortized quantum communication complexity: when many independent copies of a communication task are solved jointly, the optimal communication cost per copy converges to its QIC~\cite{Tou15}. Thus, QIC provides both an information-theoretic measure of the information flow in an interactive quantum protocol and an operational characterization of the asymptotic communication required to implement it. These properties make QIC a natural tool for studying direct-sum phenomena, and understanding the compressibility of quantum protocols.

Despite its appealing operational interpretation, QIC is substantially more difficult to analyze than its classical counterpart. At a high level, QIC is obtained by summing quantum conditional mutual information terms over the different rounds of an interactive protocol. The difficulty is that quantum conditional mutual information is itself a subtle quantity, especially when the conditioning system contains genuine quantum side information. In the classical setting, conditional mutual information admits a transparent interpretation as an average, over the conditioning variable, of ordinary mutual information. No analogous decomposition is available in general in the quantum setting: there is no canonical way to condition on a quantum register without potentially disturbing the state. This difficulty is reflected more broadly in quantum information theory, where understanding states with small quantum conditional mutual information led to sophisticated recovery-map techniques, beginning with the landmark result of Fawzi and Renner~\cite{FR15} and subsequently strengthened, simplified, and further developed in~\cite{BHOS15,BT16,SFR16}. These subtleties make it difficult to translate successful information-theoretic techniques from classical communication complexity directly to the quantum setting and constitute a central obstacle to obtaining sharp bounds on QIC.

Among concrete functions, the two-bit $\mathsf{AND}$ function is arguably the simplest nontrivial example for which IC exhibits genuinely interesting behavior. At the same time, understanding the IC of $\mathsf{AND}$ has consequences far beyond this elementary function. One of the main reasons is its fundamental connection to $\mathsf{Set} \mathsf{Disjointness}$: the disjointness problem can be viewed as applying $\mathsf{AND}$ coordinate by coordinate and determining whether any coordinate evaluates to one. $\mathsf{Set} \mathsf{Disjointness}$ is one of the central problems in communication complexity and has served as a canonical source of lower bounds, with applications to data-streaming algorithms, distributed computation, and lower bounds on extended formulations, among others~\cite{CP10,BYJKS04,BEOPV13,BM13,BP13}. Consequently, lower bounds for the IC of $\mathsf{AND}$, combined with appropriate direct-sum arguments, can be lifted to communication lower bounds for $\mathsf{Set} \mathsf{Disjointness}$. This connection has motivated a substantial line of work in the classical setting, beginning with information-theoretic lower bounds for $\mathsf{AND}$~\cite{BYJKS04} and leading to an exact characterization of its zero-error IC and sharp analyses of its behavior under error and in related multiparty settings~\cite{BGPW13,DFHL18,FHLY19}. These results in turn yield a precise understanding of the connection between the IC of $\mathsf{AND}$ and the randomized communication complexity of $\mathsf{Set} \mathsf{Disjointness}$~\cite{BGPW13,DFHL18}. In the quantum setting, Braverman, Garg, Ko, Mao, and Touchette~\cite{BGKMT15} developed an analogous approach based on QIC. In particular, they established a lower bound on the QIC of $\mathsf{AND}$ and used it as a key ingredient in proving a near-optimal lower bound for the bounded-round quantum communication complexity of $\mathsf{Set} \mathsf{Disjointness}$.

The result of the authors of~\cite{BGKMT15}, however, leaves several important questions unresolved. Quantitatively, their lower bound on the $r$-round QIC of $\mathsf{AND}$ is $\Omega(1/(r\log^8 r))$, and therefore loses polylogarithmic factors compared with the conjectured optimal dependence on the number of rounds. Moreover, their proof is indirect: rather than analyzing directly the quantum information revealed by an $\mathsf{AND}$ protocol, it proceeds through the quantum communication complexity of $\mathsf{Set} \mathsf{Disjointness}$ and ultimately relies on a threshold strong direct product theorem for quantum communication~\cite{She12}. Indeed,~\cite{BGKMT15} explicitly posed obtaining a direct information-theoretic proof for the QIC lower bound of $\mathsf{AND}$ as an open problem. 

In this work, we resolve this question by giving a direct and tight analysis of the bounded-round QIC of $\mathsf{AND}$, showing that its dependence on the number of rounds is $\Theta(1/r)$. Beyond the quantitative improvement, our proof introduces information potential, a new variational approach to QIC. Instead of manipulating quantum conditional mutual information directly, we study suitable variational representations of quantum divergences and their data-processing gaps, which allow the information revealed by individual messages to be controlled through operator inequalities and stability arguments. We believe that this perspective provides a useful new tool for the study of QIC and may find applications beyond the $\mathsf{AND}$ function.

\subsection{Main result}

Let $\mu_0$ be a uniform distribution on
$\{(0,0),(0,1),(1,0)\}$.
The following theorem establishes a lower bound on the QIC of computing $\operatorname{AND}$.
Throughout, the number of rounds counts individual messages.

\begin{theorem}\label{thm:and-qic-lower-bound}
For every
integer $r\geq 1$, every quantum communication protocol $\Pi$
that computes $\operatorname{AND}$ using at most $r$ messages
with worst-case error at most $1/3$ satisfies
\[
    \operatorname{QIC}(\Pi,\mu_0)\geq \Omega\left(\frac{1}{r}\right).
\]
The bound holds with arbitrary prior entanglement and
finite-dimensional private workspaces.
\end{theorem}
This bound is tight up to a constant factor. \cite{AA05} gives an $r$-message
protocol for $\mathrm{DISJ}_n$ with communication $O(n/r)$ when
$n\ge r^2$. Together with the DISJ-to-AND reduction
of~\cite[Lemma~4.20]{BGKMT15}, this yields an $r$-message protocol
for $\operatorname{AND}$ with worst-case error at most $1/3$ and
quantum information cost $O(1/r)$ under $\mu_0$.
Thus, the $r$-message quantum information complexity of
$\operatorname{AND}$ under $\mu_0$ is $\Theta(1/r)$.

Theorem~\ref{thm:and-qic-lower-bound}, together with the reductions
and direct-sum properties of QIC,
gives the following consequences for $\mathsf{Set} \mathsf{Disjointness}$.
Here $\mathrm{DISJ}_n(x,y)=1$ if and only if
$x_jy_j=0$ for every $j\in\{1,\ldots,n\}$.

\begin{theorem}\label{thm:main-disjointness}
For all integers $n,r\ge 1$, every quantum protocol computing $\mathrm{DISJ}_n$ with at most $r$ messages and worst-case error at most $1/3$ communicates at least $\Omega(\max\{n/r,r\})$ qubits.
\end{theorem}
In a seminal work, Bar-Yossef, Jayram, Kumar, and
Sivakumar~\cite{BYJKS04}  gave an information-theoretic proof
of the tight $\Omega(n)$ lower bound on the randomized
communication complexity of $\mathsf{Set}$ $\mathsf{Disjointness}$. Obtaining an analogous information-theoretic proof of
the tight $\Omega(\sqrt{n})$ quantum lower bound established
by Razborov~\cite{Raz02} has been a longstanding goal. In 2003, Jain, Radhakrishnan and Sen made
progress in this direction by proving an $\Omega(n/r^2)$
lower bound for $r$-round quantum protocols, which yields
an $\Omega(n^{1/3})$ lower bound without any restriction
on the number of rounds. Our $\Omega(n/r)$ bound in
Theorem~\ref{thm:main-disjointness} closes this gap:
it immediately implies the tight $\Omega(\sqrt n)$
lower bound for unrestricted interaction.
To the best of our knowledge, this gives the first
direct information-theoretic proof of the tight quantum
communication lower bound for $\mathsf{Set}$ $\mathsf{Disjointness}$.\footnote{\cite{BGKMT15} established an
$\Omega(n/(r\log^8 r))$ lower bound for $r$-round
quantum protocols. Their proof, however, relies on
Sherstov's threshold strong direct product
theorem~\cite{She12} together with a well known generalized
discrepancy lower bound $\Omega(\sqrt{n})$ for $\mathsf{Set}$ $\mathsf{Disjointness}$.
As they explicitly noted, this dependence precludes
an independent proof of the unrestricted
$\Omega(\sqrt n)$ lower bound.
In contrast, our argument does not invoke an existing
communication lower bound for $\mathsf{Set}$ $\mathsf{Disjointness}$.
} 

The following theorem establishes a direct sum theorem for bounded-round quantum communication complexity of $\mathsf{Set}$ $\mathsf{Disjointness}$.

\begin{theorem}\label{thm:direct summain-disjointness}
For all integers $n,k,r\ge 1$, every quantum protocol computing $k$ instances of
    $\mathrm{DISJ}_n$ with at most $r$ messages in total
    communicates at least $\Omega(kn/r)$ qubits, provided that,
    on every joint input, each individual answer is correct
    with probability at least $2/3$.
\end{theorem}

We further investigate asymmetric quantum communication complexity of $\mathsf{Set}$ $\mathsf{Disjointness}$ and provide an almost optimal bound. 

\begin{theorem}\label{thm:asymmain-disjointness}
For integer $n\ge 1$, every quantum protocol computing $\mathrm{DISJ}_n$
    with worst-case error at most $1/3$ satisfies
    \[
        (q_A+1)(q_B+1)\ge \Omega(n),
    \]
    where $q_A$ and $q_B$ are the total numbers of qubits
    communicated from Alice to Bob and from Bob to Alice,
    respectively. This statement imposes no bound on the
    number of messages.
    Moreover, this tradeoff is optimal up to a logarithmic factor.
\end{theorem}


\subsection{Technical Contribution}

QIC is a sum of quantum conditional mutual information terms,
whose analysis is complicated by quantum conditioning and
noncommutativity. Our approach starts instead from Petz's
resolvent representation of quantum relative entropy~\cite{Petz}.
We introduce an \emph{information potential} $V(\rho_{AB},A,B)$ for any bipartite state $\rho_{AB}$ by integrating the
resolvent divergences over a fixed parameter interval (see Definition~\ref{def:potential}). Up to a
universal rescaling, this potential lower-bounds relative entropy.
Its advantage is that each integrand in the information potential admits a concave quadratic
variational representation over operators, making its behavior
accessible through operator inequalities.

The information potential shares several nice properties with quantum relative entropy.For instance, isometries on the second register keep
the potential unchanged.  We also show that the potential is monotone under partial trace:
discarding a register restricts the variational optimization and
therefore cannot increase its value. More importantly, completing
the square expresses each resolvent data-processing gap exactly
as a weighted squared distance between the full and reduced
optimizers. To connect these gaps to QIC, we use the identity

\[
    I(A:C\mid B)_\rho
    =
    D(\rho_{ABC}\Vert\alpha_A\otimes\rho_{BC})
    -
    D(\rho_{AB}\Vert\alpha_A\otimes\rho_B),
\]
valid for any fixed positive definite reference state $\alpha_A$.
The resolvent representation and nonnegativity of its
data-processing gaps then imply
\[
    0\le
    V(\rho_{ABC},A,BC)-V(\rho_{AB},A,B)
    \le \frac{9}{2}I(A:C\mid B)_\rho.
\]
Thus, the potential preserves an essential connection to
conditional mutual information without requiring us to analyze
the latter directly.

For a safe, purified protocol, we track the potential between
Alice's input reference and Bob's registers. Receiving a message
from Alice increases the potential by a gap controlled by the
corresponding information cost; local isometries leave it
unchanged, and sending a message cannot increase it. Consequently,
QIC bounds the cumulative positive increments and hence the
net increase. Subtracting the initial value gives a progress
measure that starts at zero, without changing any message
increment.

We demonstrate that, along a purification path induced by input perturbations, the variation of the resolvent data-processing gap is bounded by the length of the purification path. Specifically, consider a fixed Alice-to-Bob message of protocol $\Pi$: when the input distribution is $\mu_1$, the purified quantum state at this stage is $\chi(\mu_1)$, and when the distribution is perturbed to $\mu_2$, the state becomes $\chi(\mu_2)$. We show that for the corresponding gaps $g(\mu_1)$ and $g(\mu_2)$ before and after tracing out the message register, the change in their square roots is bounded by a constant multiple of the integral of the rate of change of $\chi$ along the perturbation path. By squaring this variation and applying elementary inequalities, we derive a relationship between $g(\mu_1)$ and $g(\mu_2)$. Since the potential increment is the integral of $g$ over the interval $[1, 2]$, integrating both sides allows us to establish the relationship between the potential increments under distributions $\mu_1$ and $\mu_2$.

For $\mathsf{AND}$, let $\mu_0$ be uniform on its three zero
inputs and write $I_0=\operatorname{QIC}(\Pi,\mu_0)$ for an
$r$-message protocol with worst-case error at most $1/3$.
We perturb the distribution to
$\mu_p=(1-p)\mu_0+p\delta_{11}$.
The corresponding purification paths have length $O(\sqrt p)$,
so stability bounds the total potential increase by
$O(I_0+rp)$, without bounding QIC under $\mu_p$.
Conversely, in the small-information regime, local alignment
of the zero-input states and correctness on $11$ reduce the
terminal comparison to a four-dimensional calculation.
A Taylor expansion in $\sqrt p$ gives an increase of
$\Omega(\sqrt p)-O(I_0+p)$.
Comparing these bounds and choosing $p$ to be a sufficiently
small constant multiple of $r^{-2}$ yields
$I_0=\Omega(1/r)$. The logarithmic loss is avoided because
the perturbation costs only $O(p)$ per message, whereas
correctness forces a gain of order $\sqrt p$.

\subsection*{Acknowledgment}

This work is supported by the National Natural Science Foundation of China (Grant Nos. 62332009 and 12347104), the Quantum Science and Technology—National Science and Technology Major Project (Grant No. 2021ZD0302901), the NSFC/RGC Joint Research Scheme (Grant No. 12461160276), the Natural Science Foundation of Jiangsu Province (No. BK20243060), the Fundamental and Interdisciplinary Disciplines Breakthrough Plan of the Ministry of Education of China (No. JYB2025XDXM118), the "111 Center"(No. B26023), and the Fundamental Research Funds for the Central Universities (Grant no. 2026300376).

\subsection*{AI disclosure}
We used GPT-6 Astra during this work. Discussion with the model helped inspire the concept of information potential and assisted us in proving the results in this paper. The authors take full responsibility for all results and arguments presented in this paper.
\section{Preliminaries}
\label{sec:preliminaries}

\subsection{Quantum information}

For operators $M,N$ on the same Hilbert space, write
$\langle M,N\rangle=\Tr(M^\dagger N)$ for the Hilbert--Schmidt inner product.

\begin{definition}
\label{def:entropy-information}
For a density operator $\rho_{AB}$, write
$\rho_A=\Tr_B\rho_{AB}$ and $S(A)_\rho=S(\rho_A)$.
The von Neumann entropy, conditional entropy, mutual information, and
conditional mutual information are defined by
\[
\begin{aligned}
 S(\rho)&=-\Tr(\rho\log\rho),\\
 S(A\mid B)_\rho&=S(AB)_\rho-S(B)_\rho,\\
 I(A:B)_\rho&=S(A)_\rho+S(B)_\rho-S(AB)_\rho,\\
 I(A:C\mid B)_\rho
 &=S(AB)_\rho+S(BC)_\rho-S(B)_\rho-S(ABC)_\rho.
\end{aligned}
\]
The quantum relative entropy is
$D(\rho\Vert\sigma)=\Tr[\rho(\log\rho-\log\sigma)]$
when $\supp\rho\subseteq\supp\sigma$, and is $+\infty$ otherwise.
State subscripts are omitted when the state is clear.
\end{definition}

\begin{lemma}[{\cite{BGKMT15,Watrous2018}}]
\label{thm:entropy-properties}
For every state on the indicated registers,
\[
\begin{gathered}
 0\le S(A)\le\log\dim A,\qquad
 -S(A)\le S(A\mid B)\le S(A),\\
 0\le I(A:B)\le2S(A),\qquad
 0\le I(A:C\mid B)\le2S(A),\\
 I(AD:C\mid B)=I(A:C\mid B)+I(D:C\mid AB).
\end{gathered}
\]
If $\mathcal N:A\to A'$ is a quantum channel and
$\sigma_{A'BC}=(\mathcal N\otimes\id_{BC})(\rho_{ABC})$, then
$I(A':C\mid B)_\sigma\le I(A:C\mid B)_\rho$.
These information quantities are invariant under local isometries.
For a pure state on $AB$, $S(A)=S(B)$; for a pure state on $ABC$,
$S(A\mid B)=-S(A\mid C)$.
\end{lemma}

The following lemma gives an integral representation, which is known as a resolvent, for quantum relative entropy. Here we include the proof for completeness.

\begin{lemma}[Resolvent for Quantum Relative Entropy, \cite{Petz}]
\label{lemma:resolvent}
Let $\rho, \sigma \in \mathbb{C}^{d \times d}$ be strictly positive density matrices. The quantum relative entropy $D(\rho\parallel\sigma) = \mathrm{Tr}(\rho \log \rho - \rho \log \sigma)$ has the integral representation:
$$D(\rho\parallel\sigma) = \int_0^\infty \mathrm{Tr}[(\rho - \sigma)M] \frac{t}{(1+t)^2} dt$$
where $M \in \mathbb{C}^{d \times d}$ satisfies $\rho M + t M \sigma = \rho - \sigma$.
\end{lemma}
\begin{proof}
    Taking spectral decompositions of $\rho$ and $\sigma$, we have $\rho = \sum_i p_i \vert{}i\rangle\langle i\vert{}$ and $\sigma = \sum_j q_j \vert{}j\rangle\langle j\vert{}$ for $p_i, q_j > 0$.
    Projecting the equation $\rho M + t M \sigma = \rho - \sigma$ onto the basis $\{\vert{}i\rangle\langle j\vert{}\}$ directly yields the matrix elements of $M$:
    $$
    M_{ij} = \frac{p_i - q_j}{p_i + t q_j} \langle i \vert{} j \rangle
    $$
    So, the trace term $\mathrm{Tr}[(\rho - \sigma)M]$ is equal to:
    $$
    \mathrm{Tr}[(\rho - \sigma)M] = \sum_{i,j} \vert{}\langle i \vert{} j \rangle\vert{}^2 \frac{(p_i - q_j)^2}{p_i + t q_j}
    $$
    The calculus told us $\int_0^\infty \frac{(p-q)^2}{p+tq} \frac{t}{(1+t)^2} dt = p \log p - p \log q - p + q$.

    So,
    $$
    \int_0^\infty \mathrm{Tr}[(\rho - \sigma)M] \frac{t}{(1+t)^2} dt = \sum_{i,j} \vert{}\langle i \vert{} j \rangle\vert{}^2 (p_i \log p_i - p_i \log q_j - p_i + q_j)
    $$
    Since the basis states are orthonormal, the right-hand side is equal to $\Tr(\rho \log \rho) - \Tr(\rho \log \sigma) - \mathrm{Tr}(\rho) + \Tr(\sigma)$. For density matrices, $\Tr(\rho) = \Tr(\sigma) = 1$, reducing the sum exactly to $D(\rho\parallel\sigma)$.
\end{proof}

\begin{definition}
\label{def:fidelity}
For an operator $M$, let $\norm{M}_1=\Tr\sqrt{M^\dagger M}$.
The trace distance and fidelity of density operators $\rho,\sigma$ are
\[
 \Delta(\rho,\sigma)=\frac12\norm{\rho-\sigma}_1,
 \qquad F(\rho,\sigma)=\norm{\sqrt\rho\sqrt\sigma}_1.
\]
Thus $F$ denotes root fidelity. A purification of $\rho_A$ is a unit
vector $\ket{\psi}_{AE}$ such that
$\Tr_E\ket{\psi}\!\bra{\psi}=\rho_A$.
\end{definition}

\begin{lemma}[{\cite{Watrous2018,MDSFT13}}]
\label{thm:fidelity-inequalities}
For density operators $\rho,\sigma$, a quantum channel $\mathcal N$,
and a measurement effect $0\preceq E\preceq I$,
\[
\begin{gathered}
 1-F(\rho,\sigma)\le\Delta(\rho,\sigma)
 \le\sqrt{1-F(\rho,\sigma)^2},\\
 D(\rho\Vert\sigma)\ge-2\ln F(\rho,\sigma),\qquad
 \abs{\Tr E(\rho-\sigma)}\le\Delta(\rho,\sigma),\\
 \Delta(\mathcal N(\rho),\mathcal N(\sigma))\le\Delta(\rho,\sigma),\qquad
 F(\mathcal N(\rho),\mathcal N(\sigma))\ge F(\rho,\sigma).
\end{gathered}
\]
Trace distance and fidelity are invariant under applying the same
isometry to both states, and under adjoining the same independent state.
\end{lemma}

\begin{theorem}[Uhlmann's Theorem, {\cite{Watrous2018}}]
\label{thm:uhlmann}
Let $\ket{\psi}_{AE}$ and $\ket{\phi}_{AE}$ purify $\rho_A$ and
$\sigma_A$, respectively. Then
\[
 \max_{U_E}\abs{\bra{\psi}(I_A\otimes U_E)\ket{\phi}}
 =F(\rho_A,\sigma_A),
\]
where the maximum is over unitaries on $E$.
An optimizing unitary can be chosen so that
\[
 \norm{\ket{\psi}-(I_A\otimes U_E)\ket{\phi}}_2^2
 =2\bigl(1-F(\rho_A,\sigma_A)\bigr).
\]
\end{theorem}

\subsection{Quantum communication protocols}
\label{subsec:protocols}

\begin{definition}
\label{def:protocol}
Alice and Bob receive classical inputs $x\in\mathcal X$ and
$y\in\mathcal Y$ in registers $X,Y$, for finite sets $\mathcal X,\mathcal Y$.
A pure quantum protocol $\Pi$ starts with an input-independent resource
$\ket{\theta}_{A_0B_0}$, possibly entangled, and alternates local
isometries with transmissions of registers $C_1,\ldots,C_m$.
All local environments are retained in finite-dimensional private
workspaces, and Bob measures his output at the end.
Each transmission counts as one message; either party may send first.
The protocol is \emph{safe} if the original input registers are retained
unchanged and used only as controls for local operations.
If $C_i$ contains $q_i$ qubits, its communication cost is
$\QCC(\Pi)=\sum_{i=1}^m q_i$.
\end{definition}

\begin{definition}
\label{def:protocol-states}
For an input distribution $\mu$, the canonical initial purification is
\[
 \ket{\Psi_0^\mu}
 =\sum_{x,y}\sqrt{\mu(x,y)}\,
   \ket{x}_X\ket{y}_Y\ket{x}_{R_X}\ket{y}_{R_Y}
   \ket{\theta}_{A_0B_0}.
\]
The reference $R=R_XR_Y$ is never acted on.
After preparation of message $C_i$, let
$\ket{\Psi_i^\mu}_{RW_iC_iS_i}$ be the global pure state,
where $W_i$ is the sender's retained register and $S_i$ is the
receiver's register before receipt, both excluding $C_i$.
For a safe protocol, $S_i=YB_i$ on Alice-to-Bob messages and
$S_i=XA_i$ on Bob-to-Alice messages, where $A_i,B_i$ denote the
corresponding private workspaces.

For a safe protocol, write
$\ket{x}_X\ket{y}_Y\ket{\psi^{xy}}_{AB}$ for its final state on
input $(x,y)$ before measurement. Tracing out $R$ gives
\[
 \tau_{XYAB}^\mu
 =\sum_{x,y}\mu(x,y)\ket{xy}\!\bra{xy}_{XY}
   \otimes\ket{\psi^{xy}}\!\bra{\psi^{xy}}_{AB}.
\]
Classical-input information is evaluated on $\tau^\mu$; information
involving $R$ is evaluated on the coherent protocol state.
\end{definition}

\begin{definition}
\label{def:error-communication-complexity}
For a function $f:\mathcal X\times\mathcal Y\to\mathcal Z$, define the average error and the worst-case error to be 
\[
 e_\mu(\Pi,f)=\Pr_{(x,y)\sim\mu}[\Pi(x,y)\ne f(x,y)],\qquad~\mbox{and}~
 e_{\mathrm{wc}}(\Pi,f)=\max_{x,y}\Pr[\Pi(x,y)\ne f(x,y)],
\]
respectively. Let $\mathcal P_k(f,\mu,\eps)$ and $\mathcal P_k(f,\eps)$ be the
sets of pure protocols with at most $k$ messages and, respectively,
$e_\mu(\Pi,f)\le\eps$ and $e_{\mathrm{wc}}(\Pi,f)\le\eps$.
Their communication complexities are
\[
 \QCC_k^{\mathrm{avg}}(f,\mu,\eps)
 =\inf_{\Pi\in\mathcal P_k(f,\mu,\eps)}\QCC(\Pi),\qquad
 \QCC_k(f,\eps)
 =\inf_{\Pi\in\mathcal P_k(f,\eps)}\QCC(\Pi).
\]
\end{definition}

\subsection{Quantum information complexity}
\label{subsec:qic}

We use the notion of quantum information complexity as defined in \cite{Tou15,BGKMT15}.

\begin{definition}[{\cite{Tou15}}]
\label{def:qic}
For a protocol $\Pi$ and an input distribution $\mu$, define
\[
 \QIC(\Pi,\mu)=\frac12\sum_{i=1}^mI(R_XR_Y:C_i\mid S_i)_{\Psi_i^\mu}
\]
\end{definition}

\begin{proposition}[Proposition 9 of {\cite{LT17}}]
\label{prop:safe-copies}
For a classical-input protocol $\Pi$, making local safe copies of its
inputs gives a protocol $\Pi^{\mathrm{safe}}$ with the same message
count, communication cost, and output distribution, such that
$\QIC(\Pi^{\mathrm{safe}},\mu)\le\QIC(\Pi,\mu)$ for every $\mu$.
\end{proposition}

\begin{definition}
\label{def:function-qic}
For $k\ge1$ and $0\le\eps<1/2$, define
\[
\begin{aligned}
 \QIC_k^{\mathrm{avg}}(f,\mu,\eps)
 &=\inf_{\Pi\in\mathcal P_k(f,\mu,\eps)}\QIC(\Pi,\mu),\\
 \QIC_k(f,\mu,\eps)
 &=\inf_{\Pi\in\mathcal P_k(f,\eps)}\QIC(\Pi,\mu).
\end{aligned}
\]
The latter charges information under $\mu$ but requires correctness on
every input, including inputs outside $\supp\mu$.
Both infima may equivalently be restricted to safe protocols.
Omitting $k$ removes the message bound.
\end{definition}

\begin{lemma}[Lemma 1 of \cite{Tou15}]
\label{lem:qic-entropy-formula}
\[
\begin{aligned}
 0\le\QIC(\Pi,\mu)&\le\QCC(\Pi),\\
 0\le\QIC_k^{\mathrm{avg}}(f,\mu,\eps)
 &\le\QCC_k^{\mathrm{avg}}(f,\mu,\eps),\\
 0\le\QIC_k(f,\mu,\eps)&\le\QCC_k(f,\eps).
\end{aligned}
\]
The complexity inequalities also hold without a message bound.
\end{lemma}


One application of our information potential is giving the optimal $r$-round QIC of the AND function. Here we give the definition of AND.

\begin{definition}
	\label{def:and-distributions}
	The function $\AND:\{0,1\}^2\to\{0,1\}$ is
	$\AND(x,y)=x\wedge y$.
	Let $\mu_0$ be uniform on $00,01,10$, and for $p\in[0,1]$ let
	$\mu_p=(1-p)\mu_0+p\delta_{11}$, where $\delta_{11}$ is the point
	mass at $(1,1)$.
	For a fixed protocol $\Pi$, write $I_0=\QIC(\Pi,\mu_0)$.
	Define the fixed reference state
	$\alpha=\frac23\ket{0}\!\bra{0}+\frac13\ket{1}\!\bra{1}$.
	The state $\alpha$ does not vary with $p$.
\end{definition}

\begin{definition}
	\label{def:and-error}
	A protocol $\Pi$ computes $\AND$ with worst-case error at most $\eps$ if
	\[
	\Pr[\Pi(x,y)\ne\AND(x,y)]\le\eps
	\quad\text{for all }(x,y)\in\{0,1\}^2.
	\]
	The probability is over the protocol's randomness and measurements on
	the specified input. This requirement includes $11$, even though
	$\mu_0(11)=0$.
	A protocol has no false positives if it outputs $0$ with probability
	one on $00,01,10$. Its acceptance probability on $11$ is denoted by $q$.
\end{definition}

\begin{definition}
	\label{def:disjointness}
	For $x,y\in\{0,1\}^n$, set disjointness satisfies
	$\DISJ_n(x,y)=1$ if and only if $x_jy_j=0$ for every
	$j\in\{1,\ldots,n\}$.
	Protocols for $\DISJ_n$ use the message-count convention of
	Definition~\ref{def:protocol} and the worst-case-error convention of
	Definition~\ref{def:and-error}, with $\AND$ replaced by $\DISJ_n$.
\end{definition}

\section{Information Potential}
\label{sec:potential}

We introduce a new potential function, which we term the \emph{information potential}. Its definition is motivated by the integral representation of quantum relative entropy. We then give its quadratic variational representation, relate its increments to conditional mutual information and QIC, and prove stability under perturbations.

\subsection{From relative entropy to information potential}

For positive semidefinite operators $\rho,\sigma$ on the same Hilbert space satisfying
$\supp\rho\subseteq\supp\sigma$, and for $u>0$, define
\begin{equation}
\label{eq:resolvent-functional}
\mathcal F_u(\rho,\sigma)
=\left\langle\rho-\sigma,M\right\rangle,
\end{equation}
where $M$ is the solution on $\supp\sigma$ of
\[
\mathcal K_{\rho,\sigma}^{(u)}(M)=\rho M+uM\sigma=\rho-\sigma.
\]
The inner product is the Hilbert--Schmidt inner product from Section~\ref{sec:preliminaries}.
The inverse of $\mathcal K_{\rho,\sigma}^{(u)}$ acts on the space of linear operators on $\supp\sigma$. We extend $M$ to the whole Hilbert space of $\sigma$, we called $H_\sigma$, by setting it to zero on the orthogonal complement of $\supp\sigma$.

\noindent\begin{minipage}{\linewidth}
With this notation, the resolvent representation in Lemma~\ref{lemma:resolvent} reads
\[
D(\rho\Vert\sigma)
=\int_0^\infty \frac{u}{(1+u)^2}\,
\mathcal F_u(\rho,\sigma)\,du
\]
for strictly positive density operators $\rho,\sigma$.
Our potential uses the same integrand for the pair
$(\omega_{TK},\alpha_T\otimes\omega_K)$, but replaces the kernel
$u/(1+u)^2$ by $1$ and restricts the integral to $[1,2]$.
On this interval the original kernel is bounded below by $2/9$;
this comparison will connect potential increments to conditional mutual information.

\begin{definition}[Information potential]
	\label{def:potential}
	Fix a positive definite density operator $\alpha_T$ on a register $T$.
	For any state $\omega=\omega_{TK}$, define its information potential by
	\[
	V(\omega,T,K)
	=\int_1^2\mathcal F_u
	(\omega_{TK},\alpha_T\otimes\omega_K)\,du,
	\qquad \omega_K=\Tr_T\omega_{TK},
	\]
	where $\mathcal F_u$ is defined in \eqref{eq:resolvent-functional}.
\end{definition}
\end{minipage}\par

The reference state $\alpha_T$ is fixed when comparing potentials. It is not required to equal $\omega_T$.
When $\alpha_T=\omega_T$, the potential measures correlations between $T$ and $K$: it is zero exactly when $\omega_{TK}=\omega_T\otimes\omega_K$, as follows from the variational representation below. For a different reference state, it can also be positive because $\omega_T$ differs from $\alpha_T$.

\subsection{Variational representation and basic properties}

Analyzing the integrand $\langle\rho-\sigma,M\rangle$ directly is difficult because $M$ is defined implicitly by a linear equation. The useful observation is that $\mathcal F_u$ has a quadratic variational representation.
Consider the scalar analogue of our integrand. If $\rho$ and $\sigma$ were scalars, the integrand would be $(\rho-\sigma)^2/(\rho+u\sigma)$, which has the form $b^2/a$ for $a>0$. A standard variational identity gives
\[
\frac{b^2}{a}=\max_{x\in\mathbb R}(2bx-ax^2).
\]
The following lemma gives the corresponding representation for operators.

\begin{lemma}[Variational representation of information potential]
	\label{lem:potential-variational}
	For $\rho,\sigma$ as in \eqref{eq:resolvent-functional}, define
	\[
	\begin{aligned}
		\mathcal L_{\rho,\sigma}^{(u)}(M)
		=2\operatorname{Re}\Tr\bigl(M^\dagger(\rho-\sigma)\bigr)
		-\Tr(\rho MM^\dagger)-u\Tr(\sigma M^\dagger M).
	\end{aligned}
	\]
	Then
	\[
	\mathcal F_u(\rho,\sigma)
	=
	\max_M
	\mathcal L_{\rho,\sigma}^{(u)}(M).
	\]
	The maximum is over linear operators on $\supp\sigma$.
	The maximizer is the solution of $\rho M_u+uM_u\sigma=\rho-\sigma$.
	
	Substituting $\rho=\omega_{TK}$ and $\sigma=\alpha_T\otimes\omega_K$, with the maximum over operators on $\mathcal H_T\otimes\supp\omega_K$, gives
	\[
	\begin{aligned}
		V(\omega,T,K)
		&=\int_{1}^2
		\max_M
		\mathcal L_{\omega_{TK},\alpha_T\otimes \omega_K}^{(u)}(M) du
	\end{aligned}
	\]
\end{lemma}

\begin{proof}
	For every nonzero $M$,
	\[
	\begin{aligned}
		\left\langle M,\mathcal K_{\rho,\sigma}^{(u)}(M)\right\rangle=
		\Tr(\rho MM^\dagger)
		+u\Tr(\sigma M^\dagger M)>0,
	\end{aligned}
	\]
	because $\sigma$ is positive definite on $\supp \sigma$.
	Thus $\mathcal K_{\rho,\sigma}^{(u)}$ is positive definite and invertible.
	
	Set $M_u=(\mathcal K_{\rho,\sigma}^{(u)})^{-1}(\rho-\sigma)$.
	Completing the square gives
	\[
	\begin{aligned}
		\mathcal L_{\rho,\sigma}^{(u)}(M) 
		&= 2\operatorname{Re}\Tr\bigl(M^\dagger(\rho-\sigma)\bigr) - \Big[ \Tr(\rho MM^\dagger) + u\Tr(\sigma M^\dagger M) \Big]
		\\&= 
		2\operatorname{Re}\langle M, \rho-\sigma \rangle - \langle M, \mathcal K_{\rho,\sigma}^{(u)}(M) \rangle
		\\&= 2\operatorname{Re}\langle M, \mathcal K_{\rho,\sigma}^{(u)}(M_u) \rangle - \langle M, \mathcal K_{\rho,\sigma}^{(u)}(M) \rangle
		\\&=
		-\langle M-M_u, \mathcal K_{\rho,\sigma}^{(u)}(M-M_u) \rangle + \langle M_u, \mathcal K_{\rho,\sigma}^{(u)}(M_u) \rangle
		\\&=
		\mathcal F_u(\rho,\sigma)
		-
		\left\langle
		M-M_u,\,
		\mathcal K_{\rho,\sigma}^{(u)}(M-M_u)
		\right\rangle.
	\end{aligned}
	\]
	The fourth equality follows because the superoperator $\mathcal K_{\rho,\sigma}^{(u)}$
	is self-adjoint with respect to the Hilbert--Schmidt inner product.
	The subtracted term is nonnegative and vanishes exactly when
	$M=M_u$. The same identity holds for operators $M$ defined on the whole Hilbert space by extending them by zero.
\end{proof}



Completing the square also expresses the loss under a partial trace as a nonnegative quadratic form.

\begin{lemma}
	\label{lem:potential-message-gap}
	Let $\rho,\sigma$ be positive semidefinite operators with
	$\supp\rho\subseteq\supp\sigma$, and write
	$\rho^-=\Tr_C\rho$ and $\sigma^-=\Tr_C\sigma$.
	For $u>0$, let $M_u$ and $N_u^-$ be the optimizing operators for
	$(\rho,\sigma)$ and $(\rho^-,\sigma^-)$, respectively, each
	extended by zero to the whole Hilbert space. Let
	$N_u=N_u^-\otimes I_C$ and $R_u=M_u-N_u$. Then
	\[
	\begin{aligned}
		\mathcal F_u(\rho,\sigma)-\mathcal F_u(\rho^-,\sigma^-)
		&=\Tr(\rho R_uR_u^\dagger)
		+u\Tr(\sigma R_u^\dagger R_u)\geq0.
	\end{aligned}
	\]
	Consequently, for every state $\omega_{TKC}$,
	\[
	V(\omega_{TKC},T,KC)\geq V(\omega_{TK},T,K).
	\]
\end{lemma}

\begin{proof}
	By the definition of the partial trace,
	\[
	\mathcal L_{\rho,\sigma}^{(u)}(N_u^-\otimes I_C)
	=\mathcal L_{\rho^-,\sigma^-}^{(u)}(N_u^-)
	=\mathcal F_u(\rho^-,\sigma^-).
	\]
	By Lemma~\ref{lem:potential-variational}, completing the square therefore gives
	\[
		\mathcal F_u(\rho,\sigma)
		-\mathcal L_{\rho,\sigma}^{(u)}(N_u)=
		\operatorname{Tr}
		\bigl(\rho(M_u-N_u)(M_u-N_u)^\dagger\bigr)+
		u\operatorname{Tr}
		\bigl(\sigma(M_u-N_u)^\dagger(M_u-N_u)\bigr)
	\]
\end{proof}

The same variational representation shows that local isometries preserve the potential.

\begin{lemma}
	\label{lem:potential-isometries}
	The potential does not change under isometries on register $K$.
	More precisely, let $U_K:K\to K'$ be an isometry, write
	$U=I_T\otimes U_K$, and set $\widetilde\rho=U\rho U^\dagger$.
	Then $V(\widetilde\rho,T,K')=V(\rho,T,K)$.
\end{lemma}

\begin{proof}
	Put $\sigma=\alpha_T\otimes\rho_K$ and
	$\widetilde\sigma=U\sigma U^\dagger=\alpha_T\otimes\widetilde\rho_{K'}$.
	Using $U^\dagger U=I$ and the cyclicity of the trace gives
	\[
	\mathcal L_{\widetilde\rho,\widetilde\sigma}^{(u)}(UMU^\dagger)
	=\mathcal L_{\rho,\sigma}^{(u)}(M).
	\]
	The map $M\mapsto UMU^\dagger$ is a bijection between operators on
	$\supp\sigma$ and operators on $\supp\widetilde\sigma$.
	Taking the maxima gives
	$\mathcal F_u(\widetilde\rho,\widetilde\sigma)=\mathcal F_u(\rho,\sigma)$.
	Integrating over $u\in[1,2]$ proves the lemma.
\end{proof}

\subsection{Conditional mutual information and QIC}

We now apply the comparison between the two integral representations.
The nonnegative quadratic gap from Lemma~\ref{lem:potential-message-gap}
allows us to bound each potential increment by conditional mutual information.

\begin{theorem}
	\label{thm:potential-cmi}
	Let $\omega_{TKC}$ be any state and let $\alpha_T$ be a positive definite density operator. Conditional mutual information and potential have the following relation:
	\[
	\begin{aligned}
		0\leq V(\omega_{TKC},T,KC)-V(\omega_{TK},T,K)
		\leq\frac{9}{2}I(T:C\mid K)_\omega.
	\end{aligned}
	\]
\end{theorem}

\begin{proof}
	For $u>0$, write
	\[
	g_u=\mathcal F_u(\omega_{TKC},\alpha_T\otimes\omega_{KC})
	-\mathcal F_u(\omega_{TK},\alpha_T\otimes\omega_K).
	\]
	Lemma \ref{lem:potential-message-gap} gives $g_u\geq 0$.
	Lemma \ref{lemma:resolvent} shows
	\[
		D(\omega_{TKC}\Vert\alpha_T\otimes \omega_{KC})
		-D(\omega_{TK}\Vert\alpha_T\otimes \omega_{K})=
		\int_0^\infty \frac{u}{(1+u)^2}\,g_u\,du.
	\]
	Since $u/(1+u)^2\geq 2/9$ for $1\leq u\leq 2$, nonnegativity
	of $g_u$ implies
	\[
	\int_1^2 g_u\,du
	\leq \frac{9}{2}
	\bigl[
	D(\omega_{TKC}\Vert\alpha_T\otimes \omega_{KC})
	-D(\omega_{TK}\Vert\alpha_T\otimes \omega_{K})
	\bigr].
	\]
	For any state $\omega_{TQ}$,
	$$
	\begin{aligned} 
		D(\omega_{TQ} \parallel \alpha_T \otimes \omega_Q)  &= 
		-S(\omega_{TQ}) - \Tr[\omega_{TQ} \log(\alpha_T\otimes\omega_Q)] 
		\\ &= 
		-S(\omega_{TQ}) - \Tr(\omega_T \log \alpha_T) + S(\omega_Q) 
		\\ &= 
		\left(S(\omega_T) + S(\omega_Q) - S(\omega_{TQ})\right) +\Tr(\omega_T \log \omega_T-\omega_T \log \alpha_T) \\ &= I(T:Q)_\omega + D(\omega_T \parallel \alpha_T) \end{aligned}
	$$
	So
	\[
	\begin{aligned}
		&D(\omega_{TKC}\Vert\alpha_T\otimes\omega_{KC})
		-D(\omega_{TK}\Vert\alpha_T\otimes\omega_K)\\
		&=I(T:KC)_\omega-I(T:K)_\omega
		=I(T:C\mid K)_\omega.
	\end{aligned}
	\]
	Consequently,
	\[
	0\leq\int_1^2 g_u\,du
	\leq\frac{9}{2}I(T:C\mid K)_\omega.
	\]
	Notice that the proof above does not require $\omega_T=\alpha_T$.
\end{proof}

To apply this bound to a protocol, we use the following shorthand for the potential at each message position.

\begin{definition}
	\label{def:potential-protocol}
	For a pure safe protocol $\Pi$ with Bob producing the output,
	let $\rho_i$ denote the global state just after the $i$-th message has been received,
	and let $\rho_0$ be the initial state.
	Set $T=R_X$ and let $K$ consist of $R_YY$ and Bob's work registers.
	If Alice has just sent a message to Bob, $K$ includes that message. Write
	\[
	V(\rho_i)=V((\rho_i)_{TK},T,K).
	\]
\end{definition}

Since QIC is a sum of conditional mutual information terms, summing the preceding bound over Alice-to-Bob messages gives the following corollary.

\begin{corollary}
	\label{lem:potential-qic}
	Fix a safe, purified quantum protocol $\Pi$ with distribution $\mu$. Let $\mathcal I_A$ be the index set of messages sent from Alice to Bob. There exists a universal constant $C_1>0$ such that
	\[
	0\le \sum_{i\in\mathcal I_A}
	\bigl(V(\rho_i)-V(\rho_{i-1})\bigr)\le C_1 \QIC(\Pi,\mu)
	\]
\end{corollary}

\begin{proof}
	For an Alice-to-Bob message $i$, tracing $C_i$ out of the state
	after receipt gives the reduced state before receipt.
	By Theorem~\ref{thm:potential-cmi}, its potential increment satisfies
	\[
	\begin{aligned}
		0\leq V(\rho_i)-V(\rho_{i-1})
		&\leq\frac{9}{2}
		I(R_X:C_i\mid R_YYB_i)_{\Psi_i^{\mu}}\\
		&\leq\frac{9}{2}
		I(R_XR_Y:C_i\mid YB_i)_{\Psi_i^{\mu}}\\
		&=\frac{9}{2}I(R_XR_Y:C_i\mid S_i)_{\Psi_i^{\mu}}
	\end{aligned}
	\]
	
	Lemma~\ref{lem:potential-isometries} shows that local isometries on Bob's registers do not change the potential; local operations by Alice leave the reduced state unchanged. By Lemma~\ref{lem:potential-message-gap}, sending a message from Bob to Alice cannot increase the potential.
	
	Finally, summing over all Alice-to-Bob messages,
	\[
	\begin{aligned}
		0
		&\leq\sum_{i\in\mathcal I_A}
		\bigl(V(\rho_i)-V(\rho_{i-1})\bigr)\\
		&\leq\frac{9}{2}\sum_{i\in\mathcal I_A}I(R_XR_Y:C_i\mid S_i)_{\Psi_i^{\mu}}
		\leq\frac{9}{2}\sum_{i}I(R_XR_Y:C_i\mid S_i)_{\Psi_i^{\mu}}
		=9\QIC(\Pi,\mu)
	\end{aligned}
	\]
\end{proof}

\subsection{Stability under perturbations}
\label{subsec:quadratic-gap-stability}

Lemma~\ref{lem:potential-message-gap} writes the difference between $\mathcal F_u$ before and after a partial trace as a quadratic form in $M-N$. This also helps us control how this difference changes when the state is perturbed. Perturbing the input distribution is useful in communication lower bounds~\cite{JRS03,BGKMT15}. Here we compare potential increments along a path of purifications. The following lemma bounds the change in the square root of the gap by the length of that path.

\begin{lemma}
	\label{lem:quadratic-gap-stability}
	Fix a positive definite reference state $\alpha_T$.
	Let $|\chi(t)\rangle_{TKCE}$, for $t\in[t_0,t_1]$, be a
	continuously differentiable path of normalized purifications, and
	set
	\[
	\rho(t)=\Tr_E|\chi(t)\rangle\langle\chi(t)|,
	\qquad
	\sigma(t)=\alpha_T\otimes\rho_{KC}(t).
	\]
	Write $\rho^-(t)=\Tr_C\rho(t)$ and
	$\sigma^-(t)=\Tr_C\sigma(t)$, and define
	\[
	g_u(t)=\mathcal F_u(\rho(t),\sigma(t))
	-\mathcal F_u(\rho^-(t),\sigma^-(t)).
	\]
	Then, for every $u\in[1,2]$,
	\[
	\left|\sqrt{g_u(t_1)}-\sqrt{g_u(t_0)}\right|
	\leq C_{T,\alpha}\int_{t_0}^{t_1}\|\chi'(t)\|_2 dt
	\]
    where $C_{T,\alpha}$ is a constant depending only on the register $T$ and the fixed reference state $\alpha_T$.
\end{lemma}

\begin{proof}
	Let
	$
	d_T=\dim T,
	\kappa=\frac{d_T}{\lambda_{\min}(\alpha_T)},
	B=1+\sqrt\kappa.
	$
	The Schmidt decomposition and Cauchy--Schwarz give
	$|v\rangle\langle v|_{TKC}\preceq d_T I_T\otimes v_{KC}$
	for every pure state $|v\rangle$, where
	$v_{KC}=\Tr_T|v\rangle\langle v|$. Applying this to each pure state
	in a spectral decomposition of $\rho(t)$ and summing gives
	\[
	\rho(t)\preceq d_T I_T\otimes\rho_{KC}(t)
	\preceq\kappa\bigl(\alpha_T\otimes\rho_{KC}(t)\bigr)
	=\kappa\sigma(t).
	\]
	
	First, we bound the optimizing matrices. For now assume that
	$\rho(t)$ is positive definite; the general case is treated at the
	end of the proof. Fix $t$ and omit it from the notation.
	Let $M$ be the optimizer for $(\rho,\sigma)$, and set
	$A_M=I-M$ and $B_M=I+uM$.
	The equation $\rho M+uM\sigma=\rho-\sigma$ gives
	\[
	\rho B_M+uB_M\sigma=(u+1)\rho.
	\]
	Define $s=\|B_M\|_\infty$ and choose unit vectors $v,w$ such that
	$B_Mv=sw$ and $B_M^\dagger w=sv$, i.e. $v$ and $w$ are right and left singular vectors, respectively. Using Cauchy--Schwarz inequality, $|\langle x|\rho|y\rangle| \leq \sqrt{\langle x|\rho|x\rangle \langle y|\rho|y\rangle}$:
	\[
	\begin{aligned}
		\|B_M\|_\infty(\langle w|\rho|w\rangle+u\langle v|\sigma|v\rangle)
		&=(u+1)|\langle w|\rho|v\rangle|\\
		&\leq(u+1)\sqrt{\langle w|\rho|w\rangle\langle v|\rho|v\rangle}\\
		&\leq(u+1)\sqrt{\kappa \langle w|\rho|w\rangle \langle v|\sigma|v\rangle}
		\\&\leq\frac{u+1}{2\sqrt u}\sqrt\kappa(\langle w|\rho|w\rangle+u\langle v|\sigma|v\rangle).
	\end{aligned}
	\]
	The last step uses $2\sqrt{ab}\leq a+b$.
	Since $\langle w|\rho|w\rangle>0$, this proves
	\[
	\|B_M\|_\infty\leq\frac{u+1}{2\sqrt u}\sqrt\kappa,
	\qquad
	\|M\|_\infty
	\leq\frac1u+\frac{u+1}{2u\sqrt u}\sqrt\kappa
	\leq B.
	\]
	
	Next, let $Y=\rho^{1/2}A_M\sigma^{-1/2}$.
	The same equation for $M$ gives
	\[
	\rho A_M+uA_M\sigma=(u+1)\sigma,
	\qquad
	\rho A_M=B_M\sigma,
	\]
	and hence
	\[
	\rho Y+uY\sigma=(u+1)\rho^{1/2}\sigma^{1/2}.
	\]
	Choose unit vectors $v,w$ for the largest singular value of $Y$. The same calculation gives
	\[
	\begin{aligned}
		\|Y\|_\infty(\langle w|\rho|w\rangle+u\langle v|\sigma|v\rangle)
		&=(u+1)|\langle w|\rho^{1/2}\sigma^{1/2}|v\rangle|\\
		&\leq(u+1)\sqrt{\langle w|\rho|w\rangle\langle v|\sigma|v\rangle}
		\leq\frac{u+1}{2\sqrt u}(\langle w|\rho|w\rangle+u\langle v|\sigma|v\rangle).
	\end{aligned}
	\]
	Thus $\|Y\|_\infty^2\leq(u+1)^2/(4u)\leq9/8$ for
	$u\in[1,2]$. Since $B_M\sigma^{1/2}=\rho^{1/2}Y$, we obtain
	\[
	\begin{aligned}
		A_M^\dagger\rho A_M
		&=\sigma^{1/2}Y^\dagger Y\sigma^{1/2}
		\preceq\frac98\sigma,\\
		B_M\sigma B_M^\dagger
		&=\rho^{1/2}YY^\dagger\rho^{1/2}
		\preceq\frac98\rho.
	\end{aligned}
	\]
	For states with non-full support, these estimates hold by restricting
	to $\operatorname{supp}\sigma$ and extending $M$ by zero.
	Partial trace preserves $\rho\preceq\kappa\sigma$, so the same
	norm bound applies to the reduced optimizer. Lifting it by $I_C$
	does not change its norm.
	
	We now differentiate the gap.
	
	Let $M(t)$ be the maximizing matrix of $\mathcal F_u(\rho(t),\sigma(t))$, and let $N(t)$ be
	the maximizing matrix of $\mathcal F_u(\rho^-(t),\sigma^-(t))$ lifted by $I_C$. Let $R(t)=M(t)-N(t)$.
	By Lemma~\ref{lem:potential-message-gap},
	\[
	g_u=a^2+ub^2,
	\qquad
	a^2=\operatorname{Tr}\rho RR^\dagger,
	\qquad
	b^2=\operatorname{Tr}\sigma R^\dagger R
	\]
	
	Since $M,N$ maximize the variational objectives, the envelope
	theorem allows us to differentiate the objectives while keeping
	$M,N$ fixed. For $N$, we apply this before lifting it by $I_C$.
	
	So we have
	\[
	g_u'(t)=\Tr(\rho'(t) H_\rho(t))
	+\Tr(\sigma'(t) H_\sigma(t)),
	\]
	where
	\[
	\begin{aligned}
		H_\rho
		&=(M+M^\dagger-MM^\dagger)
		-(N+N^\dagger-NN^\dagger)=RA_M^\dagger+A_NR^\dagger,\\
		H_\sigma
		&=(-M-M^\dagger-uM^\dagger M)
		-(-N-N^\dagger-uN^\dagger N)=-R^\dagger B_M-B_N^\dagger R,
	\end{aligned}
	\]
	with $A_N=I-N$ and $B_N=I+uN$. Both $H_\rho$ and $H_\sigma$ are Hermitian.
	
	For a Hermitian operator $H$ and a differentiable purification, $\Tr(\rho'(t) H)=2\operatorname{Re}\langle\chi'(t)|H|\chi(t)\rangle$. Thus we bound $\|H_\rho|\chi_\rho\rangle\|_2$ and $\|H_\sigma|\chi_\sigma\rangle\|_2$ respectively. We take $|\chi_\rho(t)\rangle=|\chi(t)\rangle$, the given global purified state. We also need to choose a purification of $\sigma(t)$ with a controlled derivative. Write the original register $T$ as $T_{\mathrm{old}}$, introduce a new output register $T$, and fix a purification
	$|\alpha\rangle_{TT'}$ of $\alpha_T$.
	After a fixed rearrangement of registers, the vector
	\[
	|\chi_\sigma(t)\rangle
	=|\alpha\rangle_{TT'}\otimes
	|\chi(t)\rangle_{T_{\mathrm{old}}KCE}
	\]
	purifies $\sigma(t)$ when $T'T_{\mathrm{old}}E$ is regarded as
	the environment. The map from $|\chi(t)\rangle$ to this vector
	is a fixed isometry, so
	\[
	\|\chi_\sigma'(t)\|_2=\|\chi'(t)\|_2.
	\]
	
	For any purification $|\chi_\rho\rangle$ of $\rho$,
	the operator estimates above give
	\[
	\begin{aligned}
		\|RA_M^\dagger|\chi_\rho\rangle\|_2^2
		&=\operatorname{Tr}(A_M^\dagger\rho A_MR^\dagger R)\\
		&\leq\frac98\operatorname{Tr}(\sigma R^\dagger R)
		=\frac98b^2.
		\\
		\|A_NR^\dagger|\chi_\rho\rangle\|_2
		&\leq\|A_N\|_\infty a\leq(1+B)a.
	\end{aligned}
	\]
	Similarly, for any purification $|\chi_\sigma\rangle$ of $\sigma$,
	\[
	\begin{aligned}
		\|R^\dagger B_M|\chi_\sigma\rangle\|_2^2
		&=\operatorname{Tr}(B_M\sigma B_M^\dagger RR^\dagger)
		\leq\frac98a^2,\\
		\|B_N^\dagger R|\chi_\sigma\rangle\|_2
		&\leq\|B_N\|_\infty b\leq(1+2B)b.
	\end{aligned}
	\]
	Since $a,b\leq\sqrt{g_u}$ and $\sqrt{9/8}\leq2$, these estimates imply
	\[
	\|H_\rho|\chi_\rho\rangle\|_2\leq(3+B)\sqrt{g_u},
	\qquad
	\|H_\sigma|\chi_\sigma\rangle\|_2\leq(3+2B)\sqrt{g_u}.
	\]
	
	Set $C_{T,\alpha}=6+3B$. Using the two purification bounds therefore gives
	\[
	|g_u'(t)|
	\leq2C_{T,\alpha}\sqrt{g_u(t)}\,\|\chi'(t)\|_2.
	\]
	To handle the situation of $\sqrt {g_u}=0$, fix $\eta>0$ and differentiate
	$\sqrt{g_u+\eta}$ instead. We know
	\[
	\begin{aligned}
		\left\vert{} \sqrt{g_u(t_1) + \eta} - \sqrt{g_u(t_0) + \eta} \right\vert{} 
		&\le
		\int_{t_0}^{t_1} \left\vert{} \frac{d}{dt} \sqrt{g_u(t) + \eta} \right\vert{} dt
		\\&=
		\int_{t_0}^{t_1} \frac{\vert{}g_u'(t)\vert{}}{2\sqrt{g_u(t) + \eta}} dt
		\\&\leq 
		\int_{t_0}^{t_1} \frac{2C_{T,\alpha}\sqrt{g_u(t)}\Vert{}\chi'(t)\Vert{}_2}{2\sqrt{g_u(t) + \eta}} dt
		\\&\le 
		\int_{t_0}^{t_1} C_{T,\alpha}\Vert{}\chi'(t)\Vert{}_2 dt
	\end{aligned}
	\]
	
	Letting $\eta\rightarrow0$, we have
	\[
	\left|\sqrt{g_u(t_1)}-\sqrt{g_u(t_0)}\right|
	\leq C_{T,\alpha}\int_{t_0}^{t_1}\|\chi'(t)\|_2\,dt.
	\]

	Finally, to remove the full-rank assumption, apply the calculation to
	$\rho^\varepsilon(t)=(1-\varepsilon)\rho(t)+\varepsilon\tau$, where
	$\tau=I_{TKC}/\dim(TKC)$, and set
	$\sigma^\varepsilon(t)=\alpha_T\otimes\rho^\varepsilon_{KC}(t)$.
	Using a fixed purification of $\tau$ with an orthogonal environment
	flag gives a purification path of speed
	$\sqrt{1-\varepsilon}\|\chi'(t)\|_2$.
	The uniform bound on the optimizing operators and the variational
	formula imply continuity of $\mathcal F_u$, so letting
	$\varepsilon\to0$ proves the same estimate for the original states.
\end{proof}

The proof gives the choice
$C_{T,\alpha}=9+3\sqrt{\dim T/\lambda_{\min}(\alpha_T)}$.

\section{Application of Information Potential to AND}

Here, we apply information potential to solve the AND problem and derive the optimal QIC lower bound.

\begin{lemma}
\label{lem:potential-qic-and}
    Fix a safe, purified quantum protocol $\Pi$ with $r$ rounds for $\AND$, in which Bob outputs the answer with worst-case error at most $\varepsilon\le 1/3$.
    Write
    \[
    I_0=\operatorname{QIC}(\Pi,\mu_0).
    \]
    Let $\rho_i(0)$ be the protocol state after the $i$-th message under $\mu_0$, and let $\rho_0(0)$ be the initial state.
    Let $\mathcal I_A$ be the index set of messages sent from Alice to Bob.

    There exists a universal constant $C_1>0$ such that
    \[
    0\le \sum_{i\in\mathcal I_A}
\bigl(V(\rho_i(0))-V(\rho_{i-1}(0))\bigr)\le C_1 I_0.
    \]
\end{lemma}

\begin{proof}
This lemma is a direct consequence of Corollary~\ref{lem:potential-qic}, with $C_1=9$.
\end{proof}

\begin{lemma}
\label{lem:potential-perturbation}
    Fix a safe, purified quantum protocol $\Pi$ with $r$ rounds for $\AND$, in which Bob outputs the answer with worst-case error at most $\varepsilon\le 1/3$.
    Let $\rho_i(p)$ be the protocol state after the $i$-th message under $\mu_p$, and let $\rho_0(p)$ be the initial state.
    Let $\mathcal I_A$ be the index set of messages sent from Alice to Bob.
    There exists a universal constant $C_2>0$ such that
    for every $p\in[0,1/2]$,
    \[
    \sum_{i\in\mathcal I_A}
    \bigl(V(\rho_i(p))-V(\rho_{i-1}(p))\bigr)\le 2C_1I_0+C_2rp.
    \]
\end{lemma}

\begin{lemma}
\label{lem:potential-terminal}
    Fix a safe, purified quantum protocol $\Pi$ with $r$ rounds for $\AND$, in which Bob outputs the answer with worst-case error at most $\varepsilon\le 1/3$. Define \(\gamma_\varepsilon:=1-2\sqrt{\varepsilon(1-\varepsilon)}\).
    There exist universal constants
    $c_0,C_3,C_4,C_5>0$ and $p_0\in(0,1/2]$ such that
    if $I_0\le c_0\gamma_\varepsilon^2$, then
    for every $p\in[0,p_0]$,
    \[
    V(\rho_r(p))-V(\rho_0(p))
    \ge
    C_3\gamma_\varepsilon\sqrt p
    -
    C_4 I_0
    -
    C_5 p.
    \]
\end{lemma}

Lemmas~\ref{lem:potential-perturbation} and~\ref{lem:potential-terminal}
are proved in Sections~\ref{sec:potential-stability}
and~\ref{sec:terminal-potential-gain}, respectively.
Combining the three lemmas gives the main lower bound.

\begin{proof}[Proof of Theorem \ref{thm:and-qic-lower-bound}]
By Proposition~\ref{prop:safe-copies}, it is sufficient to consider safe pure protocols.
We may assume that $\Pi$ has exactly $r$ messages by adding trivial messages.

Write
\[
    I_0=\operatorname{QIC}(\Pi,\mu_0),
    \qquad
    \gamma=\gamma_{1/3}
    =1-2\sqrt{\frac{1}{3}\left(1-\frac{1}{3}\right)}
    =1-\frac{2\sqrt{2}}{3}>0.
\]
Let $C_1,C_2,c_0,C_3,C_4,C_5,p_0$ be the universal constants
from Lemmas~\ref{lem:potential-qic-and}, \ref{lem:potential-perturbation}, and \ref{lem:potential-terminal}.

First, local isometries preserve the potential by
Lemma~\ref{lem:potential-isometries}. A Bob-to-Alice message traces
out a register from $K$, so it cannot increase the potential by
Lemma~\ref{lem:potential-message-gap}.

It follows that, for every $p\in[0,p_0]$,
\[
\begin{aligned}
    V(\rho_r(p))-V(\rho_0(p))
    &\leq
    \sum_{i\in\mathcal I_A}
    \bigl(V(\rho_i(p))-V(\rho_{i-1}(p))\bigr)\\
    &\leq 2C_1I_0+C_2rp.
\end{aligned}
\]
$p_0\leq 1/2$ ensures that Lemma \ref{lem:potential-perturbation} applies.

If $I_0>c_0\gamma^2$, then
$I_0\geq c_0\gamma^2/r$, since $r\geq 1$.
It therefore remains to consider $I_0\leq c_0\gamma^2$.
In this case, Lemma \ref{lem:potential-terminal} gives
\[
    V(\rho_r(p))-V(\rho_0(p))
    \geq C_3\gamma\sqrt{p}-C_4I_0-C_5p
\]
for every $p\in[0,p_0]$.
Combining the two bounds yields
\[
    (2C_1+C_4)I_0
    \geq C_3\gamma\sqrt{p}-(C_2r+C_5)p.
\]

Choose
\[
    \beta=
    \min\left\{
        \sqrt{p_0},
        \frac{C_3\gamma}{2(C_2+C_5)}
    \right\}>0,
    \qquad
    p=\frac{\beta^2}{r^2}.
\]
This choice satisfies $p\leq p_0$.
Moreover, $C_2r+C_5\leq (C_2+C_5)r$, so
\[
\begin{aligned}
    (2C_1+C_4)I_0
    &\geq
    \frac{C_3\gamma\beta}{r}
    -\frac{(C_2+C_5)\beta^2}{r}\\
    &\geq \frac{C_3\gamma\beta}{2r}.
\end{aligned}
\]
Thus, in both cases,
\[
    I_0\geq \frac{c}{r},
\]
where $
    c=
    \min\left\{
        c_0\gamma^2,
        \frac{C_3\gamma\beta}{2(2C_1+C_4)}
    \right\}>0
$ is a universal constant.
\end{proof}

\section{Stability under Input Perturbations}
\label{sec:potential-stability}

This section proves Lemma \ref{lem:potential-perturbation}.
The main estimate controls the square root of a divergence gap
along a path of purifications.

\begin{proof}[Proof of Lemma \ref{lem:potential-perturbation}]
	Fix an Alice-to-Bob message position $i$.
	Use the coherent global state after preparation of the message as a purification of the reduced state on $TKC_i$.
	Here $T=R_X$ and $K=R_YYB_i$, with the message $C_i$ excluded from $K$.
	The protocol prefix is a fixed global isometry, independent of the input distribution. Since
	$\mu_p=(1-p)\mu_0+p\delta_{11}$, its coherent state can be written
	\[
	|\chi_i(t)\rangle
	=\cos t\,|\chi_i(0)\rangle
	+\sin t\,|\chi_i(11)\rangle,
	\qquad p=\sin^2t.
	\]
	The two vectors on the right are normalized and orthogonal, so
	\[
	\|\chi_i'(t)\|_2^2
	=\|{-\sin t\,\chi_i(0)+\cos t\,\chi_i(11)}\|_2^2=1.
	\]
	Thus the path from $\mu_0$ to $\mu_p$ has length $\arcsin\sqrt p$.
	Let $g_{i,u}(p)$ be the divergence gap corresponding to this
	message, obtained by tracing out $C_i$.
	
	We apply Lemma~\ref{lem:quadratic-gap-stability} with the fixed
	reference state $\alpha_T=\operatorname{diag}(2/3,1/3)$.
	Since $\dim T=2$ and $\lambda_{\min}(\alpha_T)=1/3$, the proof
	of that lemma gives
	\[
	C_{T,\alpha}
	=9+3\sqrt{\frac{\dim T}{\lambda_{\min}(\alpha_T)}}
	=9+3\sqrt6<18.
	\]
	For $0\leq p\leq1/2$, the stability estimate gives
	\[
	\begin{aligned}
		\sqrt{g_{i,u}(p)}
		&\leq\sqrt{g_{i,u}(0)}+18\arcsin\sqrt p\\
		&\leq\sqrt{g_{i,u}(0)}+36\sqrt p
	\end{aligned}
	\]
	by $\arcsin x\le 2x$. The QM-AM inequality yields
	\[
	\begin{aligned}
		g_{i,u}(p)&\le (\sqrt{g_{i,u}(0)}+36\sqrt p)^2
		\\&\le 2g_{i,u}(0)+C_2p
	\end{aligned}
	\]
	Keeping $C_i$ gives the state on $T$ and Bob's registers after receipt,
	while tracing it out gives the state before receipt: Alice's local
	preparation does not change this reduced state, and local isometries
	on Bob's registers preserve the potential.
	Integrating over $u\in[1,2]$ we have
	\[
	\begin{aligned}
		V(\rho_i(p))-V(\rho_{i-1}(p))
		&\leq2\bigl(V(\rho_i(0))-V(\rho_{i-1}(0))\bigr)
		+C_2p.
	\end{aligned}
	\]
	Summing over $i\in\mathcal I_A$:
	\[
	\begin{aligned}
		\sum_{i\in\mathcal I_A}\big(V(\rho_i(p))-V(\rho_{i-1}(p))\big)
		&\leq 2\sum_{i\in\mathcal I_A}\big((V(\rho_i(0))-V(\rho_{i-1}(0)))\big)
		+C_2rp
		\\&\le
		2C_1I_0
		+C_2rp
	\end{aligned}
	\]
\end{proof}
\section{Lower Bound of Potential Change}
\label{sec:terminal-potential-gain}

This section proves Lemma~\ref{lem:potential-terminal}.
Define $\gamma_\varepsilon=1-2\sqrt{\varepsilon(1-\varepsilon)}$.
The proof uses the final pure workspace states
$|\psi^{xy}\rangle_{AB}$, with the safe input registers
$X,Y$ written separately.

\subsection{Distance Analysis of the Final State}

To lower-bound the information potential change, we compare the purified states on inputs $00,01,10,11$. When the information cost is small, local unitaries can make the three states corresponding to inputs with correct output $0$ close to one another. Correctness, on the other hand, keeps the state on input $11$ separated from the state on input $01$. The following lemmas quantify these two observations.

\begin{lemma}[Corollary of \cite{JRS03}]
	\label{lem:holevo-local-alignment}
	Let $X$ be a classical random variable uniformly distributed over $\{0,1\}$, and let $B$ be a quantum system encoding $x \in \{0, 1\}$ as the density operator $\omega_x$. Let $|\psi_0\rangle_{AB}$ and $|\psi_1\rangle_{AB}$ be fixed purifications of $\omega_0$ and $\omega_1$, and let 
	\[
	v = I(X:B).
	\]
	There exists a unitary $U_A$ on $A$ such that
	\[
	\langle\psi_0|(U_A\otimes I_B)|\psi_1\rangle
	=F(\omega_0,\omega_1)\in[0,1],
	\qquad
	\bigl\||\psi_0\rangle
	-(U_A\otimes I_B)|\psi_1\rangle\bigr\|_2
	\leq \sqrt{2v}.
	\]
	In particular,
	$1-\langle\psi_0|(U_A\otimes I_B)|\psi_1\rangle\leq v$.
\end{lemma}

\begin{proof}
	Lemma~2 of~\cite{JRS03} gives
	\[
	1-F(\omega_0,\omega_1)
	\leq v.
	\]
	By Uhlmann's theorem, a unitary $U_A$ can be chosen so that
	the overlap of the fixed purifications is the nonnegative
	real number $F(\omega_0,\omega_1)$.
	Consequently,
	\[
	\bigl\||\psi_0\rangle
	-(U_A\otimes I_B)|\psi_1\rangle\bigr\|_2^2
	=2\bigl(1-F(\omega_0,\omega_1)\bigr)
	\leq 2v.
	\]
\end{proof}

\begin{lemma}\label{lem:alignment}
	Let $\Pi$ be a safe pure protocol. Let $X,Y$ be classical input registers jointly distributed according to $\mu_0$, which is uniform on $\{(0,0),(0,1),(1,0)\}$, and let $I_0=\QIC(\Pi,\mu_0)$.
	Write $|\psi^{xy}\rangle_{AB}$ for its final pure workspace
	state on input $(x,y)$, with the safe input registers written
	separately, and define the marginals
	\[
	\rho_B^{xy}
	=\operatorname{Tr}_A|\psi^{xy}\rangle\langle\psi^{xy}|,
	\qquad
	\rho_A^{xy}
	=\operatorname{Tr}_B|\psi^{xy}\rangle\langle\psi^{xy}|.
	\]
	Then, the conditional quantum mutual informations satisfy
	\[
	I(X:B \mid Y=0) \leq 3I_0,
	\qquad
	I(Y:A \mid X=0) \leq 3I_0.
	\]
	There exist unitaries $U_A$ on $A$ and $U_B$ on $B$ such that
	the normalized vectors
	\begin{align*}
		|\theta\rangle&=|\psi^{00}\rangle,
		&|\eta\rangle&=(I_A\otimes U_B)|\psi^{01}\rangle,\\
		|\zeta\rangle&=(U_A\otimes I_B)|\psi^{10}\rangle,
		&|\xi\rangle&=(U_A\otimes U_B)|\psi^{11}\rangle
	\end{align*}
	satisfy, with $\delta=\sqrt{6I_0}$,
	\[
	\|\eta-\theta\|_2\leq\delta,
	\qquad
	\|\zeta-\theta\|_2\leq\delta,
	\qquad
	\|\eta-\zeta\|_2\leq2\delta.
	\]
	The overlaps $\langle\theta|\eta\rangle$ and
	$\langle\theta|\zeta\rangle$ are real and nonnegative.
	In particular,
	\[
	s_0:=\langle\eta|\theta\rangle\in[0,1],
	\qquad
	1-s_0=\tfrac12\|\eta-\theta\|_2^2\leq 3I_0.
	\]
	These vectors can be realized by appending input-controlled local unitaries, which do not change QIC or information potential.
\end{lemma}

\begin{proof}
	Under the uniform distribution $\mu_0$ on $\{(0,0),(0,1),(1,0)\}$, the marginal probability of $Y=0$ is $2/3$. Conditioned on $Y=0$, the input $X$ is uniformly distributed over $\{0,1\}$. By the chain rule and data processing, Bob's final information $I(X:B\mid Y)$ is at most the sum of the incoming message information terms, and hence at most $2I_0$ under our convention for QIC. Averaging over the classical register $Y$ therefore gives
	\[
	I(X:B \mid Y=0) \leq \frac{I(X:B\mid Y)}{\Pr(Y=0)} \leq 3I_0.
	\]
	By symmetry, conditioning on $X=0$ gives $I(Y:A \mid X=0) \leq 3I_0$.
	
	We can now directly apply Lemma~\ref{lem:holevo-local-alignment}. 
	For the $Y=0$ branch, the marginal density operators on Bob's system are $\rho_B^{00}$ and $\rho_B^{10}$. Since $v = I(X:B \mid Y=0) \leq 3I_0$, Lemma~\ref{lem:holevo-local-alignment} guarantees the existence of a unitary $U_A$ on $A$ such that
	\[
	\langle\psi^{00}|(U_A\otimes I_B)|\psi^{10}\rangle = F(\rho_B^{00},\rho_B^{10}) \in [0,1],
	\]
	so this overlap is real and nonnegative. 
	Defining $|\theta\rangle = |\psi^{00}\rangle$ and $|\zeta\rangle = (U_A\otimes I_B)|\psi^{10}\rangle$, the lemma immediately gives the distance bound
	\[
	\|\zeta-\theta\|_2 \leq \sqrt{2v} \leq \sqrt{6I_0} = \delta.
	\]
	
	Similarly, for the $X=0$ branch, Bob plays the role of the correcting party. Set $v=I(Y:A\mid X=0)$. Lemma~\ref{lem:holevo-local-alignment} yields a unitary $U_B$ on $B$ defining $|\eta\rangle = (I_A\otimes U_B)|\psi^{01}\rangle$ such that $s_0 = \langle\eta|\theta\rangle \in [0,1]$, and
	\[
	1-s_0 = \tfrac{1}{2}\|\eta-\theta\|_2^2 \leq v \leq 3I_0.
	\]
	This also ensures $\|\eta-\theta\|_2 \leq \sqrt{6I_0} = \delta$.
	
	For the cross-distance $\|\eta-\zeta\|_2$, we apply the triangle inequality:
	\[
	\|\eta-\zeta\|_2 \leq \|\eta-\theta\|_2 + \|\theta-\zeta\|_2 \leq \delta + \delta = 2\delta.
	\]
	
	Finally, Alice applies $U_A$ to her local workspace only if her input is $x=1$, and Bob applies $U_B$ to his local workspace only if $y=1$. This defines $|\xi\rangle = (U_A\otimes U_B)|\psi^{11}\rangle$ on input $(1,1)$. These operations require no additional communication, so QIC does not change. Alice's operation leaves the reduced state on $TK$ unchanged, and Bob's operation preserves its potential by Lemma~\ref{lem:potential-isometries}.
\end{proof}

\begin{lemma}\label{lem:final-state-separation}
	Suppose the protocol in
	Lemma~\ref{lem:alignment} computes
	$\operatorname{AND}$ with worst-case error at most
	$0\leq\varepsilon<1/2$, and Bob produces the output.
	For the corrected workspace vectors in that lemma, let
	\[
	\gamma=1-2\sqrt{\varepsilon(1-\varepsilon)}>0,
	\qquad \omega=\langle\xi|\zeta\rangle.
	\]
	Then
	\[
	|\langle\eta|\xi\rangle|\leq1-\gamma,
	\qquad
	|\omega|\leq1-\gamma+2\delta.
	\]
	In particular,
	\[
	\|\xi-\eta\|_2\geq\sqrt{2\gamma},
	\qquad
	\|\xi-\zeta\|_2
	\geq\max\{0,\sqrt{2\gamma}-2\delta\}.
	\]
	If $I_0\leq\gamma^2/96$, then
	\[
	\operatorname{Re}\omega\leq|\omega|\leq1-\gamma/2,
	\qquad
	\|\xi-\zeta\|_2\geq\sqrt\gamma.
	\]
\end{lemma}

\begin{proof}
	Let $\omega_B^\eta=\operatorname{Tr}_A|\eta\rangle
	\langle\eta|$ and
	$\omega_B^\xi=\operatorname{Tr}_A|\xi\rangle
	\langle\xi|$.
	Inputs $01$ and $11$ have the same Bob input, so their
	corrected acceptance effect is the same operator
	$\widetilde E_1$.
	Correctness gives
	\[
	\operatorname{Tr}(\widetilde E_1\omega_B^\eta)
	\leq\varepsilon,
	\qquad
	\operatorname{Tr}(\widetilde E_1\omega_B^\xi)
	\geq1-\varepsilon.
	\]
	Consequently, their trace distance is at least $1-2\varepsilon$.
	The trace-distance--fidelity inequality and Uhlmann's theorem
	imply
	\[
	\begin{aligned}
		|\langle\eta|\xi\rangle|
		&\leq F(\omega_B^\eta,\omega_B^\xi)\\
		&\leq\sqrt{1-(1-2\varepsilon)^2}
		=2\sqrt{\varepsilon(1-\varepsilon)}=1-\gamma.
	\end{aligned}
	\]
	Since $\|\zeta-\eta\|_2\leq2\delta$, Cauchy--Schwarz gives
	\[
	|\omega|
	\leq|\langle\xi|\eta\rangle|
	+|\langle\xi|\zeta-\eta\rangle|
	\leq1-\gamma+2\delta.
	\]
	Normalization gives
	$\|\xi-\eta\|_2^2
	=2-2\operatorname{Re}\langle\eta|\xi\rangle
	\geq2\gamma$.
	The reverse triangle inequality then yields the stated
	unconditional lower bound on $\|\xi-\zeta\|_2$.
	Finally, $I_0\leq\gamma^2/96$ implies
	$\delta=\sqrt{6I_0}\leq\gamma/4$.
	Thus $\operatorname{Re}\omega\leq|\omega|\leq1-\gamma/2$, and
	$\|\xi-\zeta\|_2^2=2-2\operatorname{Re}\omega\geq\gamma$.
\end{proof}

\subsection{Analysis on the Reduced Quantum State}
\label{sec:four-dimensional}

During protocol execution, the complete quantum state contains
registers $R_XR_YXYAB$. To estimate the potential change, we first
trace out the work registers and reduce the calculation to a
four-dimensional state.

Set $T=R_X$ and $Q=R_YY$. Since the protocol is safe, the register
$Q$ is supported on the subspace spanned by $|00\rangle_{R_YY}$
and $|11\rangle_{R_YY}$. We use these two vectors as the basis
$|0\rangle_Q,|1\rangle_Q$. Thus, after also tracing out $X$,
the remaining state on $TQ$ is a $4\times4$ matrix.
For the initial and final states, write
\[
\hat\rho_j(p)=\Tr_{XAB}\rho_j(p),\qquad j\in\{0,r\}.
\]
Here $\rho_r(p)$ is the final state after the local unitaries in
Lemma~\ref{lem:alignment}. For these reduced states, we use the
abbreviation $V(\hat\rho_j(p))=V(\hat\rho_j(p),T,Q)$.
Relative to $V(\rho_j(p))$, this discards Bob's workspace $B$
from $K=QB$; the registers $X,A$ were already traced out in
the definition of the protocol potential.

For convenience, set $z=\sqrt p$ and
$a=\sqrt{(1-z^2)/3}$. To compare the states in the following
calculation, define
\[
\hat\rho_{TQ}(s,\omega,z)
=\proj0_T\otimes a^2\begin{pmatrix}1&s\\s&1\end{pmatrix}
+\proj1_T\otimes
\begin{pmatrix}
	a^2&az\omega\\
	az\overline\omega&z^2
\end{pmatrix},
\]
where $s\in[0,1]$ and $|\omega|\le1$. Its marginal and the
corresponding reference operator are
\[
\begin{aligned}
	\hat\rho_Q(s,\omega,z)
	&=\Tr_T\hat\rho_{TQ}(s,\omega,z)\\
	&=\begin{pmatrix}
		2a^2&a^2s+az\omega\\
		a^2s+az\overline\omega&a^2+z^2
	\end{pmatrix},\\
	\sigma_{TQ}(s,\omega,z)
	&=\alpha_T\otimes\hat\rho_Q(s,\omega,z).
\end{aligned}
\]
The reference state remains
$\alpha_T=\frac23\proj0+\frac13\proj1$.
With $s_0=\langle\eta|\theta\rangle$ and
$\omega=\langle\xi|\zeta\rangle$ from the preceding lemmas,
tracing out the workspace gives
\[
\hat\rho_r(p)=\hat\rho_{TQ}(s_0,\omega,\sqrt p),
\qquad
\hat\rho_0(p)=\hat\rho_{TQ}(1,1,\sqrt p).
\]
The initial state corresponds to setting both overlaps to $1$.
The following lemma allows us to lower-bound the potential change
of the complete state using these reduced states.

\begin{lemma}\label{lem:reduction}
	The potential changes of the full and reduced states satisfy
	\[
	V(\rho_r(p))-V(\rho_0(p))
	\ge V(\hat{\rho}_r(p))-V(\hat\rho_0(p)).
	\]
\end{lemma}

\begin{proof}
	By Lemma~\ref{lem:potential-message-gap}, tracing out Bob's workspace cannot increase the potential. Thus, $V(\rho_r(p))\ge V(\hat{\rho}_r(p))$. It remains to show equality for the initial state.
	
	Bob's initial workspace is independent of the input registers. Omitting the initial-state index and the parameter $p$, we therefore have $\rho_{TK}=\rho_{TQ}\otimes \rho_{B}$ and $\alpha_{T}\otimes \rho_{K}=\alpha_T\otimes \rho_{Q}\otimes \rho_{B}$.
	
	By Lemma~\ref{lem:potential-variational}, the optimizer $M$ for the reduced pair satisfies
	$$
	\rho_{TQ} M + u M (\alpha_T\otimes\rho_Q) = \rho_{TQ} - \alpha_T\otimes\rho_Q 
	$$
	
	Lift $M$ to the full state by setting $M'=M\otimes I_B$. This is an optimizer for the full pair, since
	\[
	\begin{aligned}
		\rho_{TK} M' + u M' (\alpha_T\otimes \rho_K) 
		&= (\rho_{TQ} \otimes \rho_B)(M \otimes I_B) + u (M \otimes I_B)(\alpha_T\otimes\rho_Q\otimes\rho_B)\\
		&= (\rho_{TQ} M \otimes \rho_B I_B) + u (M(\alpha_T\otimes\rho_Q) \otimes I_B\rho_B) \\
		&= (\rho_{TQ} M \otimes \rho_B) + u (M(\alpha_T\otimes\rho_Q) \otimes\rho_B)  \\
		&= (\rho_{TQ} M + u M(\alpha_T\otimes\rho_Q))\otimes\rho_B
		\\&=(\rho_{TQ} - \alpha_T\otimes\rho_Q)\otimes \rho_B
		\\&=\rho_{TK}-\alpha_T\otimes \rho_K.
	\end{aligned}
	\]
	Substituting $M'=M\otimes I_B$ into the variational objective gives the same value as for the reduced pair, because $\Tr\rho_B=1$. Integrating over $u\in[1,2]$ gives $V(\rho_0(p))=V(\hat\rho_0(p))$. Combining this equality with the final-state inequality proves the lemma.
\end{proof}

\begin{lemma}\label{lem:four-dimensional}
	For every $u\in[1,2]$, $z\in[0,1/2]$, $s\in[0,1]$, and
	$\omega\in\mathbb C$ with $|\omega|\le1$, the states defined above
	satisfy
	\[
	\begin{aligned}
		&\mathcal F_u\bigl(\hat\rho_{TQ}(s,\omega,z),
		\sigma_{TQ}(s,\omega,z)\bigr)\\
		&\quad-\mathcal F_u\bigl(\hat\rho_{TQ}(1,1,z),
		\sigma_{TQ}(1,1,z)\bigr)\\
		&\qquad\geq\frac12(1-\operatorname{Re}\omega)z
		-\frac{128}{3}(1-s)-442z^2.
	\end{aligned}
	\]
\end{lemma}

\begin{proof}
	First, we replace $s$ by $1$. By Lemma~\ref{lem:restore-overlap},
	this changes $\mathcal F_u$ by at most $\frac{128}{3}(1-s)$, so
	\[
	\begin{aligned}
		&\mathcal F_u\bigl(\hat\rho_{TQ}(s,\omega,z),
		\sigma_{TQ}(s,\omega,z)\bigr)\\
		&\qquad\geq
		\mathcal F_u\bigl(\hat\rho_{TQ}(1,\omega,z),
		\sigma_{TQ}(1,\omega,z)\bigr)
		-\frac{128}{3}(1-s).
	\end{aligned}
	\]
	
	After this replacement, the first conditional block is pure.
	Lemma~\ref{lem:pure-approximation} replaces the second block by
	another pure state, with an error of order $z^2$. For the resulting
	state, Lemma~\ref{lem:pure-ensemble} expresses $\mathcal F_u$ as a
	function $\phi_u$ of the determinant of its marginal.
	Lemmas~\ref{lem:pure-approximation} and~\ref{lem:scalar-expansion}
	then give
	\[
	\begin{aligned}
		&\Bigl|\mathcal F_u\bigl(\hat\rho_{TQ}(1,\omega,z),
		\sigma_{TQ}(1,\omega,z)\bigr)\\
		&\qquad-\phi_u(1/9)+\kappa_u z\operatorname{Re}\omega\Bigr|
		\leq221z^2,
		\qquad \kappa_u\geq\tfrac12.
	\end{aligned}
	\]
	Here $\phi_u(d)$ is the pure-state expression from Lemma~\ref{lem:pure-ensemble}, and $\kappa_u=\frac{2}{3\sqrt3}\phi_u'(1/9)$ is its linear coefficient in Lemma~\ref{lem:scalar-expansion}.
	
	Apply the same estimate with $\omega=1$ and subtract the two
	expressions. The common term $\phi_u(1/9)$ cancels, and the two
	errors add up to at most $442z^2$. Thus
	\[
	\begin{aligned}
		&\mathcal F_u\bigl(\hat\rho_{TQ}(1,\omega,z),
		\sigma_{TQ}(1,\omega,z)\bigr)\\
		&\quad-\mathcal F_u\bigl(\hat\rho_{TQ}(1,1,z),
		\sigma_{TQ}(1,1,z)\bigr)\\
		&\qquad\geq\kappa_u(1-\operatorname{Re}\omega)z-442z^2\\
		&\qquad\geq\frac12(1-\operatorname{Re}\omega)z-442z^2.
	\end{aligned}
	\]
	The last inequality follows from $1-\operatorname{Re}\omega\geq0$.
	Combining this with the first estimate proves the lemma.
\end{proof}

The coefficient of $z=\sqrt p$ is controlled by correctness through
Lemma~\ref{lem:final-state-separation}. The first error is of order
$1-s_0=O(I_0)$ by Lemma~\ref{lem:alignment}; the remaining error
is of order $z^2=p$. The appendix verifies the uniform constants
without any further protocol assumptions.

\subsection{Proof of the terminal bound}
\label{sec:terminal-proof}

\begin{proof}[Proof of Lemma~\ref{lem:potential-terminal}]
	Assume $I_0\le\gamma_\varepsilon^2/96$ and $p\in[0,1/4]$.
	Apply Lemmas~\ref{lem:reduction} and~\ref{lem:four-dimensional}
	with $s=s_0$ and $z=\sqrt p$.
	The integration interval has length one, so
	\[
	\begin{aligned}
		V(\rho_r(p))-V(\rho_0(p))
		&\ge\tfrac12(1-\operatorname{Re}\omega)\sqrt p
		-\tfrac{128}{3}(1-s_0)-442p\\
		&\ge\tfrac14\gamma_\varepsilon\sqrt p-128I_0-442p
	\end{aligned}
	\]
	The second inequality uses Lemma~\ref{lem:final-state-separation} for
	$1-\operatorname{Re}\omega$ and Lemma~\ref{lem:alignment} for $1-s_0$.
	This proves the lemma with $c_0=1/96$, $C_3=1/4$, $C_4=128$, $C_5=442$, and $p_0=1/4$.
\end{proof}

\section{Lower Bounds for Set Disjointness}
\label{sec:disjointness}

Using our lower bound for AND, we obtain several results for set disjointness, which imply quantum protocol of set disjointness is "fragile".

\subsection{Optimal Round Restriction for Set Disjointness}

Now we can give a strictly Optimal lower bound without logarithm loss for set disjointness with $r$ round interaction. 
Let
\[
    \QIC^r_0(\mathrm{AND},1/3)
    =
    \inf_{\Pi\in\mathcal P_r(\mathrm{AND},1/3)}
    \max_{\mu:\,\mu(1,1)=0}\QIC(\Pi,\mu).
\]
The following result from~\cite{BGKMT15} relates
set disjointness to AND.

\begin{theorem}[Lemma 4.20 of \cite{BGKMT15}]
\label{thm:disj-to-and}
For all integers $n,r\geq 1$,
\[
  \mathrm{QCC}_r(\mathrm{DISJ}_n,1/3)
  \geq
  n\,\mathrm{QIC}^{r}_{0}(\mathrm{AND},1/3).
\]
\end{theorem}

Together with the information lower bound for $\AND$,
this reduction gives the following communication lower bound.

\begin{theorem}
\label{thm:disj-tight}
For all integers $n,r\geq 1$,
\[
  \mathrm{QCC}_r(\mathrm{DISJ}_n,1/3)
  =
  \Omega\left(\frac{n}{r}\right).
\]
The constants are independent of $n$ and $r$.
\end{theorem}

\begin{proof}
Let $\mu_0$ be uniform on $\{00,01,10\}$.
Since $\mu_0(1,1)=0$, Theorem~\ref{thm:disj-to-and} and
Theorem~\ref{thm:and-qic-lower-bound} imply that, for a universal constant $c>0$,
\[
  \begin{aligned}
  \mathrm{QCC}_r(\mathrm{DISJ}_n,1/3)
  &\geq n\,\mathrm{QIC}^{r}_{0}(\mathrm{AND},1/3)\\
  &\geq n\,\mathrm{QIC}_r(\mathrm{AND},\mu_0,1/3)
   \geq c\,\frac{n}{r}.
  \end{aligned}
\]

Combining these two bounds proves the required lower bound.

\end{proof}

\cite{Raz02} gives
$\mathrm{QCC}(\mathrm{DISJ}_n,1/3)=\Omega(\sqrt n)$
without a message bound. Our theorem further shows that
achieving $O(\sqrt n)$ communication requires
$\Omega(\sqrt n)$ messages.

\subsection{Multiple Instances of Set Disjointness}
For a Boolean function $f$, let $f^{\times k}$ denote the task
of evaluating $f$ on $k$ input pairs. Alice receives
$\mathbf{x}=(x^{(1)},\ldots,x^{(k)})$, Bob receives
$\mathbf{y}=(y^{(1)},\ldots,y^{(k)})$, and Bob outputs
$Z=(Z_1,\ldots,Z_k)$. The error requirement is
\[
  \Pr\!\left[Z_j\neq f(x^{(j)},y^{(j)})\right]\leq\varepsilon
  \qquad
  \text{for every }(\mathbf{x},\mathbf{y})\text{ and }j\in[k].
\]
The probability is over the protocol's internal randomness
and measurements. The output errors may be correlated.
Write $\mathrm{QCC}_r(f^{\times k},\varepsilon)$
for the minimum worst-case communication cost under this
requirement, using at most $r$ individual messages in total
and allowing prior entanglement.

The following consequence of \cite[Lemma~3.10]{BGKMT15}
preserves the error guarantee for each coordinate.

\begin{lemma}
\label{lem:coordinate-qic-direct-sum}
Let $\nu$ be an input distribution for $f$, and let $\Pi$
use at most $r$ messages to compute $f^{\times k}$ with
error at most $\varepsilon$ in each coordinate on every input.
Then
\[
  \mathrm{QIC}(\Pi,\nu^{\otimes k})
  \geq k\,\mathrm{QIC}_r(f,\nu,\varepsilon).
\]
Here $\mathrm{QIC}_r(f,\nu,\varepsilon)$ requires correctness
on every input, including inputs outside the support of $\nu$.
\end{lemma}

\begin{proof}
Iterating the product-input decomposition of
\cite[Lemma~3.10]{BGKMT15} gives protocols
$\Pi_1,\ldots,\Pi_k$, each using at most $r$ messages, such that
\[
  \sum_{j=1}^k\mathrm{QIC}(\Pi_j,\nu)
  \leq \mathrm{QIC}(\Pi,\nu^{\otimes k}).
\]
On any input to $\Pi_j$, its output has the same distribution
as the $j$th output of $\Pi$ when the other input pairs are
drawn independently from $\nu$. The assumed error guarantee
therefore makes each $\Pi_j$ a worst-case
$\varepsilon$-error protocol for $f$. Each summand is at least
$\mathrm{QIC}_r(f,\nu,\varepsilon)$, proving the claim.
\end{proof}

\begin{theorem}
\label{thm:multiple-disjointness}
There exists a universal constant $c>0$ such that, for all
integers $n,k,r\geq 1$,
\[
  \mathrm{QCC}_r
  (\mathrm{DISJ}_n^{\times k},1/3)
  =\Omega\left(\frac{kn}{r}\right),
\]
with implicit constants independent of $n$, $k$, and $r$.
\end{theorem}

\begin{proof}
Let $\mu_0$ be uniform on $\{00,01,10\}$ and set
$\nu_n=\mu_0^{\otimes n}$. For any admissible protocol $\Pi$,
Lemma~\ref{lem:coordinate-qic-direct-sum},
\cite[Lemma~4.18]{BGKMT15}, and Theorem~\ref{thm:and-qic-lower-bound} give
\[
  \begin{aligned}
  \mathrm{QCC}(\Pi)
  &\geq \mathrm{QIC}(\Pi,\nu_n^{\otimes k})\\
  &\geq k\,\mathrm{QIC}_r(\mathrm{DISJ}_n,\nu_n,1/3)\\
  &\geq kn\,\mathrm{QIC}_r(\mathrm{AND},\mu_0,1/3)
   \geq c\,\frac{kn}{r}.
  \end{aligned}
\]
Taking the infimum over $\Pi$ proves the lower bound.
\end{proof}

\subsection{Almost optimal asymmetric restriction for set disjointness}

Theorem~\ref{thm:disj-tight} also constrains the communication
in each direction. For a protocol $\Pi$, let $q_A(\Pi)$ and
$q_B(\Pi)$ denote the total numbers of qubits sent from Alice
to Bob and from Bob to Alice, respectively. We can introduce the following theorem.

\begin{theorem}
\label{thm:disj-asymmetric}
There exists a universal constant $c>0$ such that, for every
integer $n\geq 1$, every quantum protocol $\Pi$ computing
$\mathrm{DISJ}_n$ with worst-case error at most $1/3$ satisfies
\[
  \bigl(q_A(\Pi)+1\bigr)\bigl(q_B(\Pi)+1\bigr)\geq \Omega(n).
\]
\end{theorem}


Thus, throughout the range $1\leq a\leq\sqrt{n}$, an
$O(a)$-qubit budget from Alice to Bob forces
$\Omega(n/a)$ qubits in the reverse direction.
In particular, $O(1)$ qubits from Alice require $\Omega(n)$
qubits from Bob, while $o(\sqrt{n})$ qubits from Alice require
$\omega(\sqrt{n})$ qubits from Bob.

The additive constants account for one-way protocols. For example, Alice can send her entire input using $n$ qubits, after which Bob computes the answer without sending any message. If both directional costs are positive, the theorem also gives $q_A(\Pi)q_B(\Pi)\geq \Omega(n)$.

\begin{proof}[Proof of Theorem~\ref{thm:disj-asymmetric}]
By Theorem~\ref{thm:disj-tight}, for a $r$ round protocol $\Pi$ solving DISJ with error less than $1/3$, we have
\[
\QCC_r(\text{DISJ}_n,1/3)=\Omega(\frac{n}{r})
\]

Since $\Pi$ exchanges $q_A(\Pi)+q_B(\Pi)$ qubits,
after removing empty messages and merging consecutive messages
from the same party, it uses at most
$2\min\{q_A(\Pi),q_B(\Pi)\}+1$ messages.
Substituting this bound on $r$, we obtain
\[
\begin{aligned}
\Omega(n)&=\left(q_A(\Pi)+q_B(\Pi)\right)\left(2\min{(q_A(\Pi),q_B(\Pi))}+1\right)
\\&\le \left(2\max{(q_A(\Pi),q_B(\Pi))}+1\right)\left(2\min{(q_A(\Pi),q_B(\Pi))}+1\right)
\\&=
\left(2q_A(\Pi)+1\right)\left(2q_B(\Pi)+1\right)
\end{aligned}
\]
So,
\[
\bigl(q_A(\Pi)+1\bigr)\bigl(q_B(\Pi)+1\bigr)\geq \Omega(n).
\]
\end{proof}

\section{An asymmetric protocol for Set Disjointness}
\label{sec:asymmetric-upper-bound}

We give an asymmetric protocol whose communication matches
our lower bound up to logarithmic factors.

\begin{theorem}\label{thm:asymmetric-upper-bound}
For every integer $1\le m\le\sqrt n$, there is a quantum
protocol computing $\mathrm{DISJ}_n$ with worst-case error
at most $1/3$ and directional communication
\[
    q_A=O(m\log(n+2)),\qquad
    q_B=O\left(
        \frac{n}{m}\log(m+1)+m\log^2(m+1)
    \right).
\]
Bob produces the final output.
\end{theorem}

In particular, for $C\log(n+2)\le a\le c\sqrt n$, choosing
$m=\Theta(a/\log(n+2))$ gives a protocol with $q_A\le a$ and
\[
    q_B=O\left(
        \frac{n\log(n+2)}{a}
        \log\left(2+\frac{a}{\log(n+2)}\right)
    \right).
\]

\begin{proof}
We first consider inputs of length $N$ with at most one
intersection. Fix an integer $k\ge2$ and let
$\theta=\pi/(2k)$. Bob prepares an $N$-qubit register $C$
in state $|0^N\rangle$. In each of $k$ cycles, he applies
\[
    U_\theta=
    \begin{pmatrix}
        \cos\theta&-\sin\theta\\
        \sin\theta&\cos\theta
    \end{pmatrix}
\]
to every coordinate with $y_i=1$, and sends $C$ to Alice.
For every coordinate with $x_i=0$, Alice swaps the
corresponding qubit with a fresh private ancilla in state
$|0\rangle$. She then sends $C$ back to Bob.

At an intersection, the rotations accumulate to $k\theta=\pi/2$,
so the corresponding qubit ends in state $|1\rangle$.
Every other coordinate is either never rotated or reset
by Alice in each cycle. Thus Bob can detect an intersection
by measuring $C$ after the final cycle.

We now compress the messages. After Alice's operation,
$C$ belongs to
\[
    \mathcal H_{\le1}
    =\operatorname{span}\{|0^N\rangle,|e_1\rangle,\ldots,|e_N\rangle\},
\]
where $|e_i\rangle$ has a single $1$ in position $i$.
This space has dimension $N+1$, so Alice can encode her
message using $\lceil\log_2(N+1)\rceil$ qubits.

Before Bob sends, each nonintersection coordinate has
excitation probability at most $\sin^2\theta=O(k^{-2})$.
In the uncompressed protocol, the Hamming weight of $C$
is therefore dominated by
$1+\operatorname{Bin}(N,\sin^2\theta)$.
Bob projects onto the subspace of Hamming weight at most
\[
    L=\min\left\{
        N,\left\lceil C_0\left(\frac{N}{k^2}+\log(k+1)\right)\right\rceil
    \right\},
\]
and encodes the resulting state in this subspace.
For a sufficiently large absolute constant $C_0$, the
discarded probability in each cycle is at most
$\eta^2/k^2$, where $\eta>0$ is a sufficiently small
constant. The gentle measurement bound and a hybrid
argument show that the total disturbance over all $k$
cycles is $O(\eta)$.

The encoding uses
\[
    O\left(
        \frac{N}{k^2}\log(k+1)+\log^2(k+1)
    \right)
\]
qubits per Bob message, by counting the strings of
Hamming weight at most $L$. Hence the subroutine costs
\[
    O(k\log(N+2))
    \quad\text{from Alice, and}\quad
    O\left(
        \frac{N}{k}\log(k+1)+k\log^2(k+1)
    \right)
    \quad\text{from Bob}.
\]
To keep these bounds on inputs with multiple intersections,
Alice also projects onto $\mathcal H_{\le1}$ before encoding.
If either projection fails, the current subroutine ends
and reports disjointness. On disjoint inputs, Alice's reset
always leaves $C$ in state $|0^N\rangle$, so the subroutine
never reports a nonexistent intersection.

To remove the promise, set
\[
    J=\lceil\log_2m\rceil+1,\qquad
    N_j=\max\{1,\lfloor n4^{-j}\rfloor\},\qquad
    k_j=\max\{2,\lceil m2^{-j}\rceil\}.
\]
For each $0\le j\le J$, use shared randomness to choose
a uniform $N_j$-element subset of the coordinates, and
run the subroutine on this subset with $k_j$ cycles.
If the intersection has size $1\le t\le m^2$, one of
these samples contains exactly one intersection with
constant probability.

Also choose a uniform subset of $\lceil n/m^2\rceil$
coordinates. Bob sends his input bits on this subset,
and Alice checks for an intersection and returns one bit.
When $t>m^2$, this test finds an intersection with constant
probability. Its communication cost is $O(n/m^2)$ qubits
from Bob and one qubit from Alice.

The protocol reports an intersection if any test finds one.
Since no test reports a nonexistent intersection, a constant
number of independent repetitions gives error at most $1/3$.
Finally, $N_j=O(n4^{-j})$ and $k_j=\Theta(m2^{-j})$, so
summing the communication over the sampling levels gives
the stated bounds. The final classical test is absorbed
in these bounds.
\end{proof}
\appendix
\section{Auxiliary Estimates for Section~\ref{sec:terminal-potential-gain}}
\label{app:terminal}

This appendix proves the estimates used in
Lemma~\ref{lem:four-dimensional}. We use the states
$\hat\rho_{TQ}(s,\omega,z)$, their marginals
$\hat\rho_Q(s,\omega,z)$, and
$\sigma_{TQ}(s,\omega,z)=\alpha_T\otimes\hat\rho_Q(s,\omega,z)$
defined in Section~\ref{sec:four-dimensional}.
As before, $z=\sqrt p$ and $a=\sqrt{(1-z^2)/3}$.
First, we bound the change when $s$ is replaced by $1$.
We then compute $\mathcal F_u$ when both conditional states are pure,
and use this formula after replacing the second block by a pure state.
Finally, we expand the resulting expression in $z$.

\subsection{A continuity estimate}

The following estimate will be used when changing the overlaps
and replacing a conditional state by a pure state.

\begin{lemma}\label{lem:optimizer-bounds}
	Let $\rho,\sigma,\widetilde\rho,\widetilde\sigma$ be density
	operators on the same finite-dimensional space, satisfying
	$\rho\preceq6\sigma$ and
	$\widetilde\rho\preceq6\widetilde\sigma$.
	For every $u\in[1,2]$,
	\begin{equation}
		\label{eq:resolvent-trace-continuity}
		\begin{aligned}
			&\bigl|\mathcal F_u(\rho,\sigma)
			-\mathcal F_u(\widetilde\rho,\widetilde\sigma)\bigr|\\
			&\qquad\leq
			24\|\rho-\widetilde\rho\|_1
			+40\|\sigma-\widetilde\sigma\|_1.
		\end{aligned}
	\end{equation}
	The two pairs need not have the same support.
\end{lemma}

\begin{proof}
	The singular-vector calculation in the proof of
	Lemma~\ref{lem:quadratic-gap-stability}, with $\kappa=6$,
	gives an operator norm bound $1+\sqrt6<4$ for both optimizing
	matrices. The calculation uses only $\rho\preceq\kappa\sigma$,
	so it applies to the pairs here.
	
	For an operator $Z$, write the
	variational objective as
	\[
	\begin{aligned}
		\mathcal L_{\rho,\sigma}^{(u)}(Z)
		&=\operatorname{Tr}
		\rho(Z+Z^\dagger-ZZ^\dagger)\\
		&\quad+\operatorname{Tr}
		\sigma(-Z-Z^\dagger-uZ^\dagger Z).
	\end{aligned}
	\]
	The zero extension of the support-restricted optimizer
	also maximizes this full-space objective.
	Indeed, the same completing-the-square identity applies
	to every operator $Z$.
	
	If $\|Z\|_\infty\leq4$, then
	\[
	\|Z+Z^\dagger-ZZ^\dagger\|_\infty
	\leq2\|Z\|_\infty+\|Z\|_\infty^2
	\leq24,
	\]
	and
	\[
	\|-Z-Z^\dagger-uZ^\dagger Z\|_\infty
	\leq2\|Z\|_\infty+u\|Z\|_\infty^2
	\leq40.
	\]
	Trace-norm duality therefore gives
	\[
	\begin{aligned}
		&\bigl|
		\mathcal L_{\rho,\sigma}^{(u)}(Z)
		-\mathcal L_{\widetilde\rho,\widetilde\sigma}^{(u)}(Z)
		\bigr|\\
		&\qquad\leq
		24\|\rho-\widetilde\rho\|_1
		+40\|\sigma-\widetilde\sigma\|_1.
	\end{aligned}
	\]
	
	Let $M$ and $\widetilde M$ be the zero-extended optimizers
	for the two pairs. Both have operator norm less than $4$.
	Optimality gives
	\[
	\begin{aligned}
		&\mathcal F_u(\rho,\sigma)
		-\mathcal F_u(\widetilde\rho,\widetilde\sigma)\\
		&\qquad\leq
		\mathcal L_{\rho,\sigma}^{(u)}(M)
		-\mathcal L_{\widetilde\rho,\widetilde\sigma}^{(u)}(M).
	\end{aligned}
	\]
	Interchanging the two pairs gives the reverse one-sided
	bound, with $\widetilde M$ as the test operator.
	Combining them proves
	\eqref{eq:resolvent-trace-continuity}.
\end{proof}

\subsection{Changing the first overlap}

We first replace $s$ by $1$ so that the first conditional state is pure.
The following lemma bounds the change in $\mathcal F_u$.

\begin{lemma}\label{lem:restore-overlap}
	For $u\in[1,2]$, $z\in[0,1/2]$, $s\in[0,1]$, and
	$|\omega|\le1$,
	\[
	\begin{aligned}
		&\bigl|\mathcal F_u(\hat\rho_{TQ}(s,\omega,z),\sigma_{TQ}(s,\omega,z))\\
		&\qquad-\mathcal F_u(\hat\rho_{TQ}(1,\omega,z),\sigma_{TQ}(1,\omega,z))\bigr|
		\le\frac{128}{3}(1-s).
	\end{aligned}
	\]
\end{lemma}

\begin{proof}
	Only the first block changes. The difference in this block has
	eigenvalues $\pm a^2(1-s)$, so
	\[
	\norm{\hat\rho_{TQ}(s,\omega,z)-\hat\rho_{TQ}(1,\omega,z)}_1
	=2a^2(1-s).
	\]
	Tracing out $T$ gives the same difference on $Q$, with the same trace
	norm. Since $\norm{\alpha_T}_1=1$, the difference between the two
	$\sigma_{TQ}$ states also has this norm.
	Both pairs satisfy the condition $\rho\preceq6\sigma$.
	By the continuity bound in Lemma~\ref{lem:optimizer-bounds}, we have
	\[
	\begin{aligned}
		&\bigl|\mathcal F_u(\hat\rho_{TQ}(s,\omega,z),\sigma_{TQ}(s,\omega,z))\\
		&\qquad-\mathcal F_u(\hat\rho_{TQ}(1,\omega,z),\sigma_{TQ}(1,\omega,z))\bigr|\\
		&\qquad\le(24+40)\,2a^2(1-s)
		\le\frac{128}{3}(1-s).
	\end{aligned}
	\]
\end{proof}

\subsection{Pure conditional states}

When each conditional state is pure, $\mathcal F_u$ can be written
using only the marginal on $Q$. For a qubit marginal, the expression
depends only on its determinant. We first give this calculation,
which will be used after replacing the second block.

\begin{lemma}\label{lem:pure-ensemble}
	Let $\widetilde\rho_{TQ}=\sum_xw_x\proj x_T\otimes P_x$, where
	$P_x=\proj{v_x}$ are pure states and $(w_x)_x$ is a probability
	distribution. Set
	$\widetilde\alpha_T=\sum_xw_x\proj x_T$ and
	$\widetilde\rho_Q=\sum_xw_xP_x$.
	For every $u>0$,
	\[
	\mathcal F_u(\widetilde\rho_{TQ},\widetilde\alpha_T\otimes\widetilde\rho_Q)
	=\frac{u+1}{u}
	\Tr\bigl[\widetilde\rho_Q(\mathbb I-\widetilde\rho_Q)
	(\mathbb I+u\widetilde\rho_Q)^{-1}\bigr].
	\]
	If $\widetilde\rho_Q$ is a qubit state and $d=\det\widetilde\rho_Q$, this equals
	\[
	\phi_u(d):=\frac{(u+1)(u+2)}{u}\,
	\frac{d}{1+u+u^2d}.
	\]
\end{lemma}

Here $\widetilde\alpha_T$ is chosen to be the marginal of
$\widetilde\rho_{TQ}$ for this calculation.
The reference state in our information potential is still
$\alpha_T$. The next lemma will also bound the change caused by
replacing $\alpha_T$ with $\widetilde\alpha_T$.

\begin{proof}
	We can omit the blocks with zero weight. By the variational
	representation, the optimization separates over the classical
	blocks: the off-diagonal blocks have no linear term and can only
	decrease the objective through the quadratic terms.
	On $\supp\widetilde\rho_Q$, the optimizer in block $x$ is
	\[
	M_x=\frac{u+1}{u}P_x(\mathbb I+u\widetilde\rho_Q)^{-1}-\frac1u\mathbb I.
	\]
	To check this, put $X_x=P_x(\mathbb I+u\widetilde\rho_Q)^{-1}$.
	Since $P_x^2=P_x$,
	$P_xX_x+uX_x\widetilde\rho_Q=P_x$.
	Substituting the expression for $M_x$ gives
	$P_xM_x+uM_x\widetilde\rho_Q=P_x-\widetilde\rho_Q$,
	which is the optimality condition.
	
	The value at the optimizer is
	$\operatorname{Re}\Tr M_x^\dagger(P_x-\widetilde\rho_Q)
	=\operatorname{Re}\Tr(P_x-\widetilde\rho_Q)M_x$.
	Using $\Tr(P_x-\widetilde\rho_Q)=0$ and cyclicity of trace gives
	\[
	\mathcal F_u(P_x,\widetilde\rho_Q)=\frac{u+1}{u}
	\Tr\bigl[P_x(\mathbb I-\widetilde\rho_Q)(\mathbb I+u\widetilde\rho_Q)^{-1}\bigr].
	\]
	Multiply by $w_x$ and sum over $x$. Since
	$\sum_xw_xP_x=\widetilde\rho_Q$, this gives the first formula.
	
	If $\widetilde\rho_Q$ is a qubit state, write its eigenvalues as
	$\lambda,1-\lambda$. Then
	\[
	\begin{aligned}
		&\Tr\bigl[\widetilde\rho_Q(\mathbb I-\widetilde\rho_Q)
		(\mathbb I+u\widetilde\rho_Q)^{-1}\bigr]\\
		&\qquad=d\left(\frac1{1+u\lambda}+\frac1{1+u(1-\lambda)}\right)
		=\frac{(u+2)d}{1+u+u^2d}.
	\end{aligned}
	\]
	Substituting this into the first formula gives $\phi_u(d)$.
	The calculation also holds when $d=0$, and all inverses are taken
	on the indicated support.
\end{proof}

\subsection{Replacing the second block}

After setting $s=1$, the first conditional state is pure.
We now make the second conditional state pure by changing its
bottom-right entry from $z^2$ to $|\omega|^2z^2$, and then normalize
the whole state. The change is of order $z^2$, so the previous lemma
can be used with only a small error.

\begin{lemma}\label{lem:pure-approximation}
	Let $0\le z\le1/2$, $|\omega|\le1$, $t=|\omega|^2$, and
	$h=1-(1-t)z^2$. Define
	\[
	\widetilde\rho_{TQ}(\omega,z)=\frac1h\left[
	\proj0_T\otimes a^2\begin{pmatrix}1&1\\1&1\end{pmatrix}
	+\proj1_T\otimes
	\begin{pmatrix}a^2&az\omega\\az\overline\omega&tz^2\end{pmatrix}
	\right],
	\]
	and write $\widetilde\rho_Q=\Tr_T\widetilde\rho_{TQ}(\omega,z)$,
	$d_\omega(z)=\det\widetilde\rho_Q$.
	For every $u\in[1,2]$,
	\[
	\bigl|\mathcal F_u(\hat\rho_{TQ}(1,\omega,z),\sigma_{TQ}(1,\omega,z))
	-\phi_u(d_\omega(z))\bigr|\le200z^2.
	\]
\end{lemma}

\begin{proof}
	Both conditional states of $\widetilde\rho_{TQ}(\omega,z)$ are pure.
	The second block is the outer product of
	$a\ket0+z\overline\omega\ket1$ before normalization.
	To bound the change in the state, we use
	\[
	\hat\rho_{TQ}(1,\omega,z)
	=h\widetilde\rho_{TQ}(\omega,z)+(1-t)z^2\proj{11}_{TQ}.
	\]
	So
	\[
	\norm{\hat\rho_{TQ}(1,\omega,z)-\widetilde\rho_{TQ}(\omega,z)}_1
	\le2(1-t)z^2\le2z^2.
	\]
	Tracing out $T$ gives
	$\norm{\hat\rho_Q(1,\omega,z)-\widetilde\rho_Q}_1\le2z^2$.
	
	To apply Lemma~\ref{lem:pure-ensemble}, we also need to replace
	$\alpha_T$ by $\widetilde\alpha_T=\Tr_Q\widetilde\rho_{TQ}(\omega,z)$.
	The two probabilities in $\widetilde\alpha_T$ are
	\[
	\widetilde w_0=\frac{2a^2}{h},\qquad
	\widetilde w_1=\frac{a^2+tz^2}{h}.
	\]
	Since $h\ge3/4$, they satisfy
	$1/2\le\widetilde w_0\le2/3$ and
	$1/3\le\widetilde w_1\le1/2$. Moreover,
	\[
	\norm{\widetilde\alpha_T-\alpha_T}_1
	=\frac{4tz^2}{3h}\le\frac{16}{9}z^2.
	\]
	Combining the changes on $T$ and $Q$, we have
	\[
	\begin{aligned}
		\norm{\sigma_{TQ}(1,\omega,z)
			-\widetilde\alpha_T\otimes\widetilde\rho_Q}_1
		&\le\norm{\hat\rho_Q(1,\omega,z)-\widetilde\rho_Q}_1
		+\norm{\alpha_T-\widetilde\alpha_T}_1\\
		&\le\frac{34}{9}z^2.
	\end{aligned}
	\]
	The original pair satisfies $\rho\preceq6\sigma$. The replacement pair also satisfies
	this condition, since
	$\widetilde\rho_{TQ}(\omega,z)\preceq2\mathbb I_T\otimes\widetilde\rho_Q$
	and $\widetilde\alpha_T\succeq\mathbb I_T/3$.
	Thus Lemma~\ref{lem:optimizer-bounds} bounds the change in
	$\mathcal F_u$ by
	\[
	24(2z^2)+40\left(\frac{34}{9}z^2\right)
	=\left(48+\frac{1360}{9}\right)z^2\le200z^2.
	\]
	Finally, Lemma~\ref{lem:pure-ensemble} gives
	\[
	\mathcal F_u(\widetilde\rho_{TQ}(\omega,z),
	\widetilde\alpha_T\otimes\widetilde\rho_Q)
	=\phi_u(d_\omega(z)).
	\]
	So the stated bound includes both the change in the state and the
	change from $\sigma_{TQ}(1,\omega,z)$ to
	$\widetilde\alpha_T\otimes\widetilde\rho_Q$.
\end{proof}

\subsection{The linear coefficient and the uniform remainder}

It remains to expand $\phi_u(d_\omega(z))$ around $z=0$.
We first estimate the change in the determinant, and then substitute
it into the scalar function $\phi_u$.

\begin{lemma}\label{lem:scalar-expansion}
	For $d_\omega(z)$ in Lemma~\ref{lem:pure-approximation}, set
	$d_0=1/9$ and $\ell=2/(3\sqrt3)$.
	For $u\in[1,2]$, $z\in[0,1/2]$, and $|\omega|\le1$,
	\[
	\bigl|\phi_u(d_\omega(z))-\phi_u(d_0)
	+\kappa_u z\operatorname{Re}\omega\bigr|\le21z^2,
	\]
	where
	\[
	\kappa_u=\ell\phi_u'(d_0)
	=\frac{2(u+1)^2(u+2)}{3\sqrt3\,u(1+u+u^2/9)^2}
	\ge\frac12.
	\]
\end{lemma}

\begin{proof}
	Summing the two blocks gives
	\[
	\widetilde\rho_Q=\frac1h
	\begin{pmatrix}
		2a^2&a^2+az\omega\\
		a^2+az\overline\omega&a^2+tz^2
	\end{pmatrix},\qquad
	d_\omega(z)=\frac{a^4-2a^3z\operatorname{Re}\omega+a^2tz^2}{h^2}.
	\]
	First, we estimate the determinant. The numerator differs from
	$d_0-\ell z\operatorname{Re}\omega$ by at most $z^2$, because
	\[
	\begin{aligned}
		|a^4-1/9|&\le\tfrac29z^2,\\
		2z\left|a^3-\tfrac1{3\sqrt3}\right|
		&\le\tfrac1{\sqrt3}z^3\le\tfrac1{2\sqrt3}z^2,\\
		a^2tz^2&\le\tfrac13z^2,
	\end{aligned}
	\]
	and $2/9+1/(2\sqrt3)+1/3<1$.
	The middle inequality uses
	$1-(1-z^2)^{3/2}\le\frac32z^2$.
	For the denominator, we have
	\[
	h^{-2}\le\frac{16}{9},\qquad
	0\le h^{-2}-1\le\frac{32}{9}z^2,\qquad
	|d_0-\ell z\operatorname{Re}\omega|\le\frac13.
	\]
	Combining these estimates gives
	\[
	|d_\omega(z)-d_0+\ell z\operatorname{Re}\omega|
	\le\frac{80}{27}z^2\le3z^2.
	\]
	
	Next, we estimate the scalar function. Since
	$\widetilde\rho_Q$ is a qubit state,
	$d_\omega(z),d_0\in[0,1/4]$.
	Direct differentiation gives
	\[
	\phi_u'(d)=\frac{(u+1)^2(u+2)}{u(1+u+u^2d)^2},\qquad
	\phi_u''(d)=-\frac{2u(u+1)^2(u+2)}{(1+u+u^2d)^3}.
	\]
	On the stated ranges,
	$0<\phi_u'(d)\le3$ and $|\phi_u''(d)|\le6$.
	Also,
	$|d_\omega(z)-d_0|\le\ell z+3z^2\le2z$.
	By Taylor's theorem between $d_0$ and $d_\omega(z)$, we get
	\[
	\begin{aligned}
		&\bigl|\phi_u(d_\omega(z))-\phi_u(d_0)
		+\ell\phi_u'(d_0)z\operatorname{Re}\omega\bigr|\\
		&\hspace{2em}\le3(3z^2)+\tfrac12\,6(2z)^2=21z^2.
	\end{aligned}
	\]
	
	Finally, we give a lower bound on the linear coefficient. We have
	\[
	\phi_u'(d_0)
	=\frac{u+2}{u}\left(1+\frac{u^2}{9(u+1)}\right)^{-2}
	\ge2\left(\frac{27}{31}\right)^2,
	\]
	because $(u+2)/u\ge2$ and $u^2/(9(u+1))\le4/27$.
	Multiplying by $\ell$ gives
	$\kappa_u\ge972/(961\sqrt3)>1/2$.
\end{proof}

\bibliographystyle{alpha}
\bibliography{references}

\end{document}